\documentclass[12pt,reqno]{amsart}

\usepackage{fullpage}
\usepackage{graphicx}
\usepackage{caption}
\usepackage{subcaption}
\usepackage{float}
\usepackage{amssymb, amsmath, amsxtra}

\newtheorem{thm}{Theorem}[section]
\newtheorem{cor}{Corollary}[section]
\newtheorem{lem}{Lemma}[section]

\newtheorem{prop}{Proposition}[section]

\numberwithin{equation}{section}
\def\Rnum{{\mathbb R}}
\def\sgn{{\rm sgn}}
\def\arctanh{{\rm arctanh}}

\def\Hop{{\mathcal{H}}}
\def\Eop{{\mathcal{E}}}
\def\Dop{{\mathcal{D}}}
\def\E{{\rm E}}
\def\tbinom#1#2{{\textstyle\binom{#1}{#2}}}

\def\parder#1{\partial_{#1}}

\def\const{\text{const.}}

\begin{document}

\allowdisplaybreaks[3]

\title{A general family of multi-peakon equations\\ and their properties}

\author{
Stephen C. Anco$^1$
\lowercase{\scshape{and}}
Elena Recio$^{1,2}$ 
\\\\\lowercase{\scshape{
${}^1$Department of Mathematics and Statistics\\
Brock University\\
St. Catharines, ON L2S3A1, Canada}} \\\\
\lowercase{\scshape{
${}^2$Department of Mathematics\\
Faculty of Sciences, University of C\'adiz\\
Puerto Real, C\'adiz, Spain, 11510}}\\
}

\begin{abstract}
A general family of peakon equations is introduced, 
involving two arbitrary functions of the wave amplitude and the wave gradient. 
This family contains all of the known breaking wave equations, 
including the integrable ones:
Camassa-Holm equation, Degasperis-Procesi equation, Novikov equation, 
and FORQ/modified Camassa-Holm equation.
One main result is to show that all of the equations in the general family 
possess weak solutions given by multi-peakons
which are a linear superposition of peakons with time-dependent amplitudes and positions. 
In particular, 
neither an integrability structure nor a Hamiltonian structure is needed
to derive $N$-peakon weak solutions for arbitrary $N>1$. 
As a further result, 
single peakon travelling-wave solutions are shown to exist under a simple condition on one of the two arbitrary functions in the general family of equations, 
and when this condition fails, 
generalized single peakon solutions that have a time-dependent amplitude 
and a time-dependent speed are shown to exist. 
An interesting generalization of the Camassa-Holm and FORQ/modified Camassa-Holm equations 
is obtained by deriving the most general subfamily of peakon equations that possess 
the Hamiltonian structure shared by the Camassa-Holm and FORQ/modified Camassa-Holm equations. 
Peakon travelling-wave solutions and their features, 
including a variational formulation (minimizer problem),
are derived for these generalized equations. 
A final main result is that $2$-peakon weak solutions are investigated 
and shown to exhibit several novel kinds of behaviour, 
including the formation of a bound pair consisting of a peakon and an anti-peakon 
that have a maximum finite separation. 
\end{abstract}

\maketitle

\begin{center}
emails: 
sanco@brocku.ca, 
elena.recio@uca.es
\end{center}

\section{Introduction}

The study of 
dispersive nonlinear partial differential equations (PDEs) of the form
\begin{equation}\label{gen-fam}
u_t -u_{txx} = F(u,u_x,u_{xx}) + G(u,u_x,u_{xx})u_{xxx}
\end{equation}
for $u(t,x)$
has attracted growing attention in the last two decades, 
because particular equations in this family describe breaking waves 
and some of these equations are integrable systems. 

One of the first well-studied equations 
is the Camassa-Holm (CH) equation
\begin{equation}\label{CH}
u_t -u_{txx} = - 3u u_x + 2 u_x u_{xx} + u u_{xxx} 
\end{equation}
which was shown in 1993 to arise \cite{Cam-Hol,Cam-Hol-Hym} 
from the theory of shallow water waves 
by an asymptotic expansion of Euler's equations for inviscid fluid flow. 
The CH equation provides a model of wave breaking 
for a large class of solutions
in which the wave slope $u_x$ blows up in a finite time
while the wave amplitude $u$ remains bounded 
\cite{Con,ConEsc-1998,ConEsc-1998b,ConEsc-2000}. 
Solitary travelling wave solutions of the CH equation consist of \cite{Len}
peaked waves $u=a\exp(-|x-ct|)$, with the amplitude-speed relation $a=c$, 
where the wave slope $u_x$ is discontinuous at the wave peak. 
These solutions are known as peakons \cite{Cam-Hol,Alb-Cam-Hol-Mar,Cao-Hol-Tit}, 
and interactions of an arbitrary number of peakons are described by 
multi-peakon solutions \cite{Bea-Sat-Szm-1999,Bea-Sat-Szm-2000}
which have the form of a linear superposition of peaked waves 
with time-dependent amplitudes and speeds. 
Peakons are not classical solutions but instead are weak solutions 
\cite{ConStr,ConMol-2000,ConMol-2001} 
satisfying an integral formulation of the CH equation. 
More remarkably, 
the CH equation is an integrable system 
in the sense that it possesses \cite{Cam-Hol,Fis-Sch,Fuc-Fok} 
a Lax pair, a bi-Hamiltonian structure, 
and an infinite hierarchy of symmetries and conservation laws. 

The CH equation was the only known peakon equation until 2002, 
when the Degasperis-Procesi (DP) equation
\begin{equation}\label{DP}
u_t -u_{txx} = - 4u u_x +3 u_x u_{xx} + u u_{xxx}
\end{equation}
was derived by the method of asymptotic integrability \cite{Deg-Pro,Deg-Hol-Hon}. 
This equation arises from the asymptotic theory of shallow water waves \cite{Dul-Got-Hol}
and possesses peakon and multi-peakon solutions \cite{Deg-Hol-Hon,Lud-Szm-2003,Lud-Szm-2005,Szm-Zhou}, 
as well as shock-type solutions \cite{Lun}. 
It is also an integrable system, possessing \cite{Deg-Hol-Hon} 
a Lax pair, a bi-Hamiltonian structure, 
and an infinite hierarchy of symmetries and conservation laws. 

In 1995, 
the Fokas-Fuchssteiner-Olver-Rosenau-Qiao (FORQ) equation 
\begin{equation}\label{FORQ}
u_t -u_{txx} = (-u^3+uu_x^2+u^2u_{xx}-u_x^2u_{xx})_x,
\end{equation}
was derived by applying a bi-Hamiltonian splitting method 
\cite{Fok95a,Olv-Ros,Fok-Olv-Ros}
to the modified Korteweg-de Vries (mKdV) equation. 
At the same time, 
this equation also appeared in other work \cite{Fuc} 
on applications of Hamiltonian methods
and was shown to arise from the asymptotic theory of surface water waves \cite{Fok95b}.
The FORQ equation \eqref{FORQ} is an integrable system. 
Its bi-Hamiltonian structure was obtained in Ref.~\cite{Olv-Ros,Qia-Li},
and its Lax pair along with its single peakon solutions appears in Ref.~\cite{Sch,Qia}. 
Multi-peakon solutions have been considered in 2013 in Ref.~\cite{Gui-Liu-Olv-Qu}. 

Similarly to the relationship between the mKdV equation and the ordinary KdV equation,
the FORQ equation and the CH equation share one of their two Hamiltonian structures 
in common \cite{Gui-Liu-Olv-Qu}
and are related by a combined Miura-Liouville  transformation \cite{Kan-Liu-Olv-Qu}. 
For these reasons and other similarities, 
the FORQ equation is also called the modified Camassa-Holm (mCH) equation 
in recent literature. 

Following those discoveries of integrable peakon equations, 
a major direction of work was to find other integrable equations in the family \eqref{gen-fam}. 
In 2009, 
a classification of integrable polynomial generalizations of the CH equation
with quadratic and cubic nonlinearities was obtained \cite{Nov}, 
which produced several new integrable equations.  
The most interesting of these is the Novikov (N) equation
\begin{equation}\label{N}
u_t -u_{txx} = - 4u^2 u_x +3 u u_x u_{xx} + u^2 u_{xxx} .
\end{equation}
Its peakon and multi-peakon solutions have been derived recently \cite{Wan-Hon,Hon-Lun-Szm}. 

In addition to the integrable equations in the family \eqref{gen-fam}, 
there are many non-integrable equations that admit peakons and multi-peakons.  
Two examples are the $b$-family \cite{Deg-Hol-Hon,Hol-Hon}
\begin{equation}\label{b-fam}
u_t -u_{txx} = - (b+1) u u_x + b u_x u_{xx} + u u_{xxx},
\quad
b \neq 0
\end{equation}
which includes both the CH equation \eqref{CH} ($b=2$) and the DP equation \eqref{DP} ($b=3$), 
and the modified $b$-family \cite{Mi-Mu}
\begin{equation}\label{mb-fam}
u_t -u_{txx} = - (b+1) u^2 u_x + b u u_x u_{xx} + u^2 u_{xxx},
\quad
b \neq 0
\end{equation}
which includes the N equation \eqref{N} ($b=3$). 
All of the equations in the $b$-family with $1<b\leq 3$ 
exhibit wave breaking \cite{Gui-Liu-Tian} similarly to the CH equation. 

Very recently, 
two larger unified families of equations have been investigated.
The first family \cite{Gra-Him,Anc-Sil-Fre,Him-Man-2016a}
\begin{equation}\label{bp-fam}
u_t -u_{txx} = - (b+1) u^p u_x + b u^{p-1} u_x u_{xx} + u^p u_{xxx},
\quad
b \neq 0, 
p \neq 0
\end{equation}
contains the CH equation \eqref{CH} ($b=2$, $p=1$), 
the DP equation \eqref{DP} ($b=3$, $p=1$),
and the N equation \eqref{N} ($b=3$, $p=2$), 
but not the mCH equation \eqref{FORQ}.
The second family is given by \cite{Him-Man-2016b}
\begin{equation}\label{kabc-fam}
\begin{aligned}
u_t -u_{txx} & = - (b+1) u^p u_x +(9a+b+2c-3p)u^{p-2}u_x^3
+ (3p-2c) u^{p-1} u_x u_{xx} 
\\&\qquad
- 6a u^{p-2} u_x u_{xx}^2
+ (u^p-3au^{p-2}u_x^2) u_{xxx},
\quad
p \neq 0
\end{aligned}
\end{equation}
which includes the first family \eqref{bp-fam} when $a=0$ and $b+2c=3p$,
and which also includes the mCH equation \eqref{FORQ} when $a=1/3$, $c=1$, and $p=2$. 
Every equation in this larger unified family \eqref{kabc-fam} 
possesses single peakon solutions, 
but multi-peakon solutions are admitted only when 
the nonlinearity power is related to the coefficients by $6a+b+2c=3p$
\cite{Him-Man-2016b}. 
This condition includes all of the equations in the first unified family \eqref{bp-fam},
whose multi-peakon solutions were derived independently in Refs.~\cite{Gra-Him,Anc-Sil-Fre}. 

No equations in either of these families \eqref{bp-fam} and \eqref{kabc-fam}
apart from the CH, DP, N, and mCH equations
are believed to be integrable. 

Another family of equations generalizing the CH equation \eqref{CH} 
is given by 
\begin{equation}\label{pgCH}
u_t -u_{txx} - \left(\tfrac{1}{2} p u^{p-1}(u^2 -u_x^2) +u^p(u- u_{xx})\right)_x = 0,
\quad
p\neq 0 . 
\end{equation}
This peakon equation was derived \cite{Anc-Rec-Gan-Bru}
by generalizing one of the Hamiltonian structures of the CH equation, 
in a search for analogs of the family of 
generalized Korteweg de Vries (gKdV) equations \cite{Fuc}
\begin{equation}\label{gKdV}
u_t -u^p u_x - u_{xxx} =0,
\quad
p \neq 0 . 
\end{equation}
The gKdV equation reduces to the ordinary KdV equation when $p=1$ 
and shares one of its two Hamiltonian structures; 
the same relationship holds between the CH equation \eqref{CH}
and the generalized equation \eqref{pgCH}.  
However, 
while this generalized CH equation possesses 
single peakon solutions \cite{Anc-Rec-Gan-Bru}, 
it does not admit multi-peakon solutions,
as will be seen from the results below. 

In the present paper, 
we are motivated to look for the most general nonlinear dispersive wave equations that belong to the family \eqref{gen-fam}
and possess single peakon and multi-peakon solutions. 
Our starting point is the observation that all of the known multi-peakon equations 
share a similar general form when they are expressed as evolution equations 
in terms of the momentum variable
\begin{equation}
m=u-u_{xx}
\end{equation}
which plays an important role in wave breaking phenomena. 
Specifically, the family of equations 
\begin{equation}\label{fg-fam}
m_t + f(u,u_x)m + (g(u,u_x)m)_x =0
\end{equation}
contains the CH equation \eqref{CH}, DP equation \eqref{DP}, mCH equation \eqref{FORQ}, Novikov equation \eqref{N}, 
and their different unifications \eqref{bp-fam} and \eqref{kabc-fam} that admit multi-peakon solutions, 
as shown in Table~\ref{table:peakon_eqns}. 

\begin{table}[htb]
\centering
\caption{Peakon equations}
\label{table:peakon_eqns}
\begin{tabular}{cc|c}
\hline
$f(u,u_x)$ & $g(u,u_x)$ & Name
\\
\hline
\hline
$u_x$
& 
$u$
&
Camassa-Holm 
\\
$2u_x$
&
$u$
&
Degasperis-Procesi 
\\
$uu_x$
& 
$u^2$
&
Novikov
\\
$0$
& 
$u^2-u_x^2$
&
Fokas-Olver-Rosenau-Qiao/
\\
&&
modified Camassa-Holm 
\\
\hline
\hline
$(b-1)u_x$
&
$u$
&
$b$-family
\\
$(b-2)uu_x$
&
$u^2$
&
modified $b$-family
\\
$(b-p)u^{p-1}u_x$
&
$u^p$
&
unified CH-DP-N family
\\
$u^{p-3}u_x( (b-p)u^2$
&
$u^{p-2}( u^2 -3a u_x^2 )$
&
CH-DP-N-mCH 
\\
$+3a(p-2) u_x^2 )$
%$(2-p)(b+2c-3p)u_x^2 \big)$
&
%$(b+2c-3p)u_x^2 \big)$
&
unified family
\\
\hline
\end{tabular}
\end{table}

We will call equation \eqref{fg-fam} the {\em $fg$-family}
and study its properties when $f(u,u_x)$ and $g(u,u_x)$ are arbitrary non-singular functions.
As a main result, 
in section~\ref{sec:peakons}, 
$N$-peakon solutions for all $N\geq 2$ 
will be derived as weak solutions for the entire $fg$-family, 
and single peakon travelling-wave solutions 
will be shown to exist under a simple condition on $f(u,u_x)$.
Interestingly, 
when this condition fails, 
more general single peakon solutions that have a time-dependent amplitude 
and a time-dependent speed are shown to exist. 

The importance of these results is that they show 
multi-peakons exist for nonlinear dispersive wave equations
without using or relying on any integrability properties or any Hamiltonian structure. 
This is a sharp contrast to the situation for multi-soliton solutions of evolution equations,
where the existence of an arbitrary number $N>2$ of solitons usually requires
that the evolution equation be integrable,
and finding the $N$-soliton solution usually involves explicit use of the integrability structure (e.g., an inverse scattering transform or a bilinear formulation). 

In section~\ref{sec:hamiltonian-struct}, 
the most general form of $f(u,u_x)$ and $g(u,u_x)$ is obtained 
such that the multi-peakon equation \eqref{fg-fam} admits the Hamiltonian structure 
shared by the CH and mCH equations. 
This yields a large family of Hamiltonian multi-peakon equations 
\begin{equation}
m_t + u_x f_1(u^2-u_x^2)  m +((u f_1(u^2-u_x^2)  + g_1(u^2-u_x^2))m)_x =0
\end{equation} 
which involves two arbitrary functions $f_1$ and $g_1$ of $u^2-u_x^2$. 
Several properties of this family are studied. 
First, 
this family is shown to have conserved momentum, energy, and $H^1$ norm
for all classical solutions $u(t,x)$,
and conservation laws for weak solutions are discussed. 
Second, 
all single peakon solutions $u=a\exp(-|x-ct|)$ are derived 
and shown to have a speed-amplitude relation $c=c(a)$ 
which is nonlinear whenever $f_1$ or $g_1$ are non-constant functions. 
This leads to an interesting relation between the properties of 
peakons with $a>0$ and anti-peakons with $a<0$. 
Third, a minimizer principle is obtained for these peakon solutions, 
which provides a starting point for establishing their stability. 

In section \ref{sec:examples},
one-parameter subfamilies of the CH-mCH Hamiltonian family are explored, 
with $f_1$ and $g_1$ taken to be general powers of $u^2-u_x^2$. 
These subfamilies represent nonlinear generalizations of the CH equation and mCH equation given by 
\begin{equation}\label{gCH}
m_t  + (u^2-u_x^2)^{p-1}u_xm + (u(u^2-u_x^2)^{p-1} m)_x = 0,
\quad
p\geq 1
\end{equation}
and 
\begin{equation}\label{gmCH}
m_t +((u^2-u_x^2)^p m)_x = 0,
\quad
p\geq 1 . 
\end{equation}
A unified generalization of these subfamilies 
\begin{equation}\label{gCHmCH}
m_t + a u_x(u^2-u_x^2)^{k/2} m + ( a u(u^2-u_x^2)^{k/2} m + b (u^2-u_x^2)^{(k+1)/2} m )_x =0,
\quad
k\geq 0
\end{equation}
is also discussed. 
Each of these multi-peakon equations share the common Hamiltonian structure 
admitted by the CH and mCH equations,
and consequently they describe an analog of the well-known Hamiltonian family of 
generalized KdV equations. 

The effect of higher-power nonlinearities on interactions of peakons
is explored in section~\ref{sec:interactions}
by studying the behaviour of 2-peakon weak solutions 
for the generalized CH and mCH equations \eqref{gCH} and \eqref{gmCH}
when $p=2$, which is compared to the behaviour in the ordinary case $p=1$. 
As main results, 
for both of these generalized peakon equations, 
qualitatively new behaviours are shown to occur in the interaction of two peakons 
whose amplitudes have opposite signs, 
namely, an ordinary peakon and an anti-peakon. 

In the case of the $p=2$ generalized mCH equation \eqref{gmCH}, 
the peakon and anti-peakon can form a bound pair
which has a maximum finite separation in the asymptotic past and future. 
The pair evolves by slowly collapsing, 
such that a collision occurs in a finite time, 
followed by asymptotically expanding, 
with the amplitudes being finite for all time. 

In the case of the $p=2$ generalized CH equation \eqref{gCH}, 
the peakon and anti-peakon can exhibit a finite time blow-up in amplitude, 
before and after they undergo a collision. 
Starting at the collision, their separation increases to a finite maximum 
and then decreases to a limiting non-zero value when the blow-up occurs. 

Neither of these types of behaviour have been seen previously 
in interactions of peakon weak solutions. 
This indicates that peakons can have a wide variety of interesting interactions 
for different multi-peakon equations in the general $fg$-family \eqref{fg-fam},
and that the form of the nonlinearity in these equations has a large impact on 
how peakons can interact. 

Some concluding remarks are made in section~\ref{sec:remarks}.

\section{Peakon solutions}\label{sec:peakons}

In the analysis of all of the peakon equations listed in Table~\ref{table:peakon_eqns}, 
single and multi peakons are commonly derived as weak solutions
\cite{ConStr,ConMol-2000,ConMol-2001,Gui-Liu-Tian,Hak,Gui-Liu-Olv-Qu}
in the setting of an integral formulation. 
We will now show that the $fg$-family \eqref{fg-fam} 
has a weak formulation for general functions $f(u,u_x)$ and $g(u,u_x)$, 
subject to mild conditions, 
and then we will use this formulation to obtain 
single peakon and multi-peakon solutions.

\subsection{Weak formulation}

A weak solution $u(t,x)$ of equation \eqref{fg-fam}
is a distribution that satisfies an integral formulation of the equation
in some suitable function space. 
We will derive this formulation by the usual steps for wave equations \cite{Eva}. 
First, we multiply equation \eqref{fg-fam} 
by a test function $\psi(t,x)$ 
(which is smooth and has compact support), 
and integrate over $-\infty < x < \infty$ and $0\leq t < \infty$. 
Next, we integrate by parts to remove all derivatives of $u$ 
higher than first-order. 
There are three different terms to consider. 
The first one is given by 
\begin{equation}\label{weak-term1}
\int_{0}^{\infty}\int_{-\infty}^{\infty} \psi (u_t-u_{txx})\,dx \,dt
= \int_{0}^{\infty}\int_{-\infty}^{\infty} (\psi - \psi_{xx}) u_t \,dx \,dt . 
\end{equation}
The next one consists of 
\begin{equation}\label{weak-term2}
\begin{aligned}
\int_{0}^{\infty}\int_{-\infty}^{\infty} \psi (u-u_{xx})f(u,u_x) \,dx \,dt
& = \int_{0}^{\infty}\int_{-\infty}^{\infty} \psi ( u f - D_xF +u_x F_u ) \,dx \,dt
\\
& = \int_{0}^{\infty}\int_{-\infty}^{\infty} ( \psi (u f +u_x F_u)+\psi_x F )\,dx \,dt , 
\end{aligned}
\end{equation}
where
\begin{equation}\label{F}
F(u,u_x) = \int f(u,u_x) \, d u_x . 
\end{equation}
The final one is, similarly, 
\begin{equation}\label{weak-term3}
\begin{aligned}
\int_{0}^{\infty}\int_{-\infty}^{\infty} \psi ((u-u_{xx})g(u,u_x))_x \,dx \,dt
& = -\int_{0}^{\infty}\int_{-\infty}^{\infty} \psi_x (u-u_{xx})g \,dx \,dt
\\
& = -\int_{0}^{\infty}\int_{-\infty}^{\infty} ( \psi_x (ug +u_x G_u) +\psi_{xx} G )\,dx \,dt , 
\end{aligned}
\end{equation}
where 
\begin{equation}\label{G}
G(u,u_x) = \int g(u,u_x) \, d u_x . 
\end{equation}
Then, combining these three terms \eqref{weak-term1}, \eqref{weak-term2}, \eqref{weak-term3}, 
we obtain the integral (weak) equation 
\begin{equation}\label{pde_weak}
0  = \int_{0}^{\infty}\int_{-\infty}^{\infty} \big( 
\psi( u_t + f u + F_u u_x) 
-\psi_x( g u + G_u u_x -F) 
- \psi_{xx} (G  +u_t)
\big) \, dx \, dt . 
\end{equation}

An equivalent weak formulation can be obtained by putting $\phi=\psi-\psi_{xx}$
in the integral (weak) equation \eqref{pde_weak}. 
Note that $\psi=\kappa(x)*\phi$ is given by a convolution integral using the kernel 
\begin{equation}\label{kernel}
\kappa(x) = \tfrac{1}{2}\exp(-|x|),
\quad
(1-\partial_x^2)^{-1} = \kappa(x)* 
\end{equation}
for the inverse of the operator $1-\partial_x^2$. 
Then we have 
\begin{equation}\label{u_pde_weak}
0  = \int_{0}^{\infty}\int_{-\infty}^{\infty} \big( 
\phi (u_t+G +{\kappa *(f u + F_u u_x-G)}) + {\kappa_x *(g u + G_u u_x -F)}
\big) \, dx \, dt
\end{equation}
where $\phi(t,x)$ is a test function. 

The weak equations \eqref{pde_weak} and \eqref{u_pde_weak} 
will hold for all distributions $u(t,x)$ in a suitable function space
whenever $f(u,u_x)$ and $g(u,u_x)$ satisfy some mild regularity/growth conditions. 
In particular, 
the following result is straightforward to prove. 

\begin{prop}\label{weak_conds}
If both $|f(u,u_x)|\leq C_1|u|^k+C_2|u_x|^k$ 
and $|g(u,u_x)|\leq C_1|u|^k+C_2|u_x|^k$ 
hold for some positive integer $k$, 
then the $fg$-equation \eqref{fg-fam} 
has a weak formulation \eqref{pde_weak},
or equivalently \eqref{u_pde_weak}, 
for $u(t,x)\in C^1(\Rnum^+; L_{\text{loc}}^1(\Rnum))\cap C^0(\Rnum^+; W_{\text{loc}}^{1,k+1}(\Rnum))$. 
\end{prop}

The weak equation \eqref{pde_weak} will be necessary 
for the derivation of multi-peakon solutions of the $fg$-equation \eqref{fg-fam}, 
but single-peakon solutions can be derived more simply by starting from 
a reduction of the $fg$-equation to a travelling wave ODE.

\subsection{Single peakon solutions}

A single peakon is a travelling wave of the form 
\begin{equation}\label{singlepeakon}
u(t,x)= a \exp(-|x-ct|)
\end{equation}
where the wave speed  $c$ is an arbitrary constant
while the peak amplitude $a$ 
is related to $c$ in a specific way that depends on the nonlinearity in the $fg$-equation \eqref{fg-fam}. 
As the wave slope is discontinuous at the wave peak, 
a peakon is a weak solution rather than a classical solution of the $fg$-equation. 

Single peakon solutions can also be derived from the reduction of the $fg$-equation
to a travelling-wave ODE, 
arising through a combined time and space translation symmetry 
that leaves $x-ct$ invariant. 
Specifically, the travelling wave ansatz
\begin{equation}\label{travellingwave}
u = U(\xi), 
\quad 
\xi = x-ct, 
\quad 
c = \const \neq 0
\end{equation}
reduces the $fg$-equation \eqref{fg-fam} to the ODE
\begin{equation}\label{tw_ode}
(U - U'')f(U,U') + ((U - U'')( g(U,U') -c ))' = 0 . 
\end{equation}
A peakon solution will not be a classical solution of this ODE 
but instead will be a weak solution. 
To obtain the weak formulation of this ODE \eqref{tw_ode},
corresponding to the weak formulation \eqref{pde_weak} of the $fg$-equation, 
we multiply the ODE by a test function $\psi$
(which is smooth and has compact support), 
integrate over $-\infty < \xi < \infty$,
and use integration by parts to leave at most first-order derivatives of $U$ in the integral. 
This yields 
\begin{equation}\label{ode_weak}
0 = \int_{-\infty}^{\infty} \left( 
 \psi ( Uf + U' (F_{U}-c) ) 
+ \psi'( F - U g  - U' G_{U} )
+ \psi''( c U' -G )
\right) \, d\xi
\end{equation}
where, now, 
\begin{equation}
F(U,U') = \int f \, d U', 
\quad
G(U,U') = \int g \, d U' . 
\end{equation}
Just as in Proposition~\ref{weak_conds}, 
only mild conditions on the functions $f$ and $g$ 
are needed for this weak equation \eqref{ode_weak} 
to hold for all $U(\xi)$ in a suitable function space. 

To find the single peakon solutions, 
we substitute a general peakon expression 
\begin{equation}
U = a \exp(-|\xi|),
\quad 
a = \const
\end{equation}
into the weak equation \eqref{ode_weak}, 
and we split up the integral into the intervals 
$(-\infty,0)$ and $(0,\infty)$.
Integrating by parts, and using the relations 
\begin{gather}
U' = 
\begin{cases}
U, &  \xi <0\\
-U, & 0<\xi
\end{cases}
\label{tw_U_rels}\\
U(0) = a,
\quad
U'(0^\pm) = \mp a
\label{tw_U_jump_rels}
\end{gather}
which hold on each interval, 
we obtain the following result. 

\begin{thm}\label{thm:single_peakon}
The weak formulation \eqref{ode_weak} of the travelling wave ODE \eqref{tw_ode}
possesses single peakon travelling-wave solutions \eqref{singlepeakon} 
if and only if
\begin{equation}\label{peakon_cond_G}
c= \frac{G(a,a) - G(a,-a)}{2a} \neq 0
\end{equation}
and 
\begin{equation}\label{peakon_cond_F}
F(a, a) = F(a,-a) 
\end{equation}
where $a$ is the amplitude and $c$ is the wave speed.
These two conditions \eqref{peakon_cond_F}--\eqref{peakon_cond_G}
hold if and if only the coefficient functions $f$ and $g$ in the $fg$-equation \eqref{fg-fam} 
have the respective forms
\begin{equation}\label{fg_1peakon}
\begin{aligned}
f(u,u_x) & = \int_{u}^{u_x} f_1(u,y)\, dy + u_x f_1(u,u_x) +f_2(u,u_x) , 
\\
g(u,u_x) & = g_0(u) +\int_{u}^{u_x} g_1(u,y)\, dy + u_x g_1(u,u_x) + g_2(u,u_x) , 
\end{aligned} 
\end{equation}
for some function $g_0$ of $u$, and some functions $f_1,f_2,g_1,g_2$ of $u$ and $u_x$, 
such that 
\begin{equation}\label{fg_1peakon'}
g_0(u)\neq 0, 
\quad
f_i(u,u_x)+ f_i(u,-u_x)= 0,
\quad
g_i(u,u_x)+ g_i(u,-u_x)= 0,
\quad
i=1,2 . 
\end{equation}
In particular, the speed-amplitude relation for single peakon travelling waves is given by 
\begin{equation}\label{fg_speed_ampl_rel}
c = g_0(a) . 
\end{equation}
\end{thm}

\begin{proof}
It will be useful to introduce the notation
\begin{equation}
h_+ = h(U,U')|_{\xi<0} = h(U,U),
\quad
h_- = h(U,U')|_{\xi>0} = h(U,-U)
\end{equation}
for any function $h(U,U')$. 

To proceed, 
we first rearrange the terms in the weak equation \eqref{ode_weak}:
\begin{equation}\label{ode_weak_rearranged}
0 = \int_{-\infty}^{\infty} \left( 
c(\psi''- \psi) U' 
+ (\psi' F + \psi F_{U}U') 
- (\psi'' G  + \psi' G_{U} U')
+ (\psi f -\psi'g) U 
\right) \, d\xi . 
\end{equation}
Now we consider each term separately. 

Integration by parts twice on the first term in equation \eqref{ode_weak_rearranged}, 
combined with the relations \eqref{tw_U_rels}--\eqref{tw_U_jump_rels}, 
yields 
\begin{equation}\label{term1_weak_ode}
\begin{aligned}
\int_{-\infty}^{\infty} c(\psi''- \psi) U' \, d\xi 
& = \int_{-\infty}^{0} c(\psi'' - \psi) U \, d\xi -\int_{0}^{\infty} c(\psi'' - \psi) U \, d\xi 
\\
& = 2 c \psi'(0)U(0) + c\left( -\int_{-\infty}^{0} (\psi'+ \psi) U \, d\xi + \int_{0}^{\infty} (-\psi'+ \psi) U \, d\xi \right) 
\\
& = 2 a c \psi'(0) . 
\end{aligned}
\end{equation}

Next, the second term in equation \eqref{ode_weak_rearranged}
gives
\begin{equation}
\begin{aligned}
\int_{-\infty}^{\infty} (\psi' F + \psi F_{U}U') \, d\xi 
& = \int_{-\infty}^{0} (\psi' F_+ + \psi (F_{U})_+ U) \, d\xi  +\int_{0}^{\infty} (\psi' F_- - \psi (F_{U})_- U) \, d\xi 
\\
& = \psi(0) (F_+ - F_-)|_{\xi = 0} 
-\int_{-\infty}^{0} \psi \Big(\frac{dF_+}{d\xi}-(F_{U})_+ U\Big) \, d\xi 
\\&\qquad
-\int_{0}^{\infty} \psi \Big(\frac{dF_-}{d\xi} +(F_{U})_- U\Big) \, d\xi 
\end{aligned}
\end{equation}
after the use of the relations \eqref{tw_U_rels}, 
with 
\begin{equation}
\frac{dF_\pm }{d\xi} 
=  \pm U ( F_{U} \pm F_{U'} )_\pm  
= \pm U(F_{U})_\pm   +U f_\pm . 
\end{equation}
Hence, we have
\begin{equation}\label{term2_weak_ode}
\int_{-\infty}^{\infty} (\psi' F + \psi F_{U}U') \, d\xi 
= \psi(0) (F_+ - F_-)|_{\xi = 0} 
-\int_{-\infty}^{\infty } \psi U f \, d\xi . 
\end{equation}
The third term in equation \eqref{ode_weak_rearranged}
similarly yields
\begin{equation}\label{term3_weak_ode}
-\int_{-\infty}^{\infty} (\psi'' G  + \psi' G_{U} U') \, d\xi 
= -\psi'(0) (G_+-G_-)|_{\xi = 0} 
+\int_{-\infty}^{\infty } \psi' U g \, d\xi . 
\end{equation}

Last, the fourth term in equation \eqref{ode_weak_rearranged}
combines with the previous terms \eqref{term1_weak_ode}, \eqref{term2_weak_ode} and \eqref{term3_weak_ode}, 
giving 
\begin{equation}
(2a c -(G_+ - G_-)|_{\xi=0}) \psi'(0)  +(F_+ - F_-)|_{\xi=0} \psi(0) = 0 . 
\end{equation}
This equation is satisfied for all test functions $\psi$ if and only if
$2a c -(G_+ - G_-)|_{\xi=0}=0$
and $(F_+ - F_-)|_{\xi=0}= 0$. 
Substituting the relations \eqref{tw_U_jump_rels} into these two conditions,
we obtain equations \eqref{peakon_cond_G}--\eqref{peakon_cond_F}. 

The form \eqref{fg_1peakon} for the functions $f$ and $g$ is obtained 
from the conditions \eqref{peakon_cond_G}--\eqref{peakon_cond_F} as follows. 
First, 
we decompose $F$ and $G$ into even and odd parts under reflection $(u,u_x)\to (u,-u_x)$,
by expressing 
$F=u_x F_1 + F_2$ and $G=u_x G_1 + G_2$
where $F_1,F_2,G_1,G_2$ are reflection-invariant functions of $u,u_x$. 
Conditions \eqref{peakon_cond_G}--\eqref{peakon_cond_F} 
then become $F_1(a,a)=0$ and $G_1(a,a)\neq 0$,
which hold iff $F_1(u,u)=0$ and $G_1(u,u)\neq 0$.
Next, from relations \eqref{F} and \eqref{G}, we have
$f= F_{u_x} = F_1 + u_xF_{1u_x} +F_{2u_x}$
and 
$g= G_{u_x} = G_1 + u_xG_{1u_x} +G_{2u_x}$. 
We now put $f_1= F_{1u_x}$ and $f_2= F_{2u_x}$, 
and likewise $g_1= G_{1u_x}$ and $g_2= G_{2u_x}$, 
all of which are odd in $u_x$. 
This yields 
$F_1 = \int_{u}^{u_x} f_1\, du_x$ and $G_1= \int_{u}^{u_x} g_1\, du_x + g_0(u)$,
using $F_1(u,u)=0$ and $G_1(u,u)\neq 0$. 
Hence, $f,g$ are given by expressions \eqref{fg_1peakon}--\eqref{fg_1peakon'}, 
which completes the proof. 
\end{proof}

Theorem~\ref{thm:single_peakon} establishes that 
single peakon travelling-wave solutions \eqref{singlepeakon} exist 
for a large class of equations in the $fg$-family \eqref{fg-fam}. 
In particular, 
sufficient conditions are that $f(u,u_x)$ is odd in $u_x$,
and that $g(u,u_x)+ g(u,-u_x)$ is non-vanishing in $u$,
whereby $F(u,u_x)$ is an even function of $u_x$,
and $G(u,u_x)-G(u,-u_x)=g_0(u)$ is a non-vanishing function. 

This theorem also shows that there is no restriction on the possible form of 
the speed-amplitude relation $c=g_0(a)$ for single peakon travelling-waves, 
since $g_0(u)$ can be an arbitrary function of $u$. 
If $g_0(u)$ is identically zero, then single peakon travelling-waves will be stationary, 
$c=0$. 
Consequently, 
it is more natural to regard the peak amplitude $a$ as an arbitrary constant, 
with the speed $c$ then being given in terms of $a$ by the speed amplitude relation. 
To distinguish the situations when $a$ is positive versus negative, 
the travelling wave solution \eqref{singlepeakon} is typically called 
an ordinary peakon when $a$ is positive, 
and an anti-peakon when $a$ is negative. 
The speed of an anti-peakon compared to the corresponding ordinary peakon 
depends on the reflection symmetry properties of the function $g_0(a)$: 
specifically, $c_+=g_0(a)$ with $a>0$ will be equal to $c_-=g_0(-a)$ with $-a<0$ 
if and only if $g_0(a)$ is an even function of $a$. 
This shows that peakons and anti-peakons with the same absolute amplitude $|a|$ 
will have different speeds, $c_+\neq c_-$, 
when (and only when) $g_0(a)$ is not invariant under reflections. 

Even more interestingly, 
when $f(u,u_x)$ and $g(u,u_x)$ fail to satisfy 
the necessary and sufficient conditions \eqref{peakon_cond_G}--\eqref{peakon_cond_F},
a generalized peakon solution still exists 
but with both the speed and the amplitude being time-dependent, 
as we will show from the derivation of multi-peakon solutions.

\subsection{Multi-peakon solutions}

A multi-peakon solution is 
a linear superposition of peaked travelling waves 
given by
\begin{equation}\label{multi-peakon}
u(t,x) = \sum_{i=1}^{N} \alpha_i(t) \exp(-|x-\beta_i(t)|), 
\quad N \geq 2
\end{equation}
with time-dependent amplitudes $\alpha_i(t)$ and positions $\beta_i(t)$.

We now derive $N$-peakon solutions for the $fg$-equation \eqref{fg-fam},
for all $N\geq 1$, with $f(u,u_x)$ and $g(u,u_x)$ being arbitrary non-singular functions. 
For convenience, 
we use the notation
\begin{equation}\label{u_expr}
u = \sum_i \alpha_i e^{-|x_i|},
\quad 
x_i = x - \beta_i
\end{equation}
where the summation is understood to go from 1 to $N$.
The $x$-derivatives of $u$ are distributions given by 
\begin{equation}\label{ux_expr}
u_x = -\sum_i \alpha_i e^{-|x_i|} \sgn(x_i)
\end{equation}
and
\begin{equation}\label{uxx_expr}
u_{xx} = \sum_i \alpha_i (e^{-|x_i|} -2\delta(x_i))
\end{equation}
in terms of the sign function
\begin{equation}\label{sgn}
\sgn(x) = \pm 1,
\quad
\pm x>0 
\end{equation}
and the Dirac delta distribution
\begin{equation}\label{dirac}
\delta(x) = \frac{d}{dx}\left(\tfrac{1}{2}\sgn(x)\right) . 
\end{equation}
Similarly, the $t$-derivatives of $u$ are given by the distributions
\begin{equation}\label{ut_expr}
u_t = \sum_i \big( \dot\alpha_i e^{-|x_i|} +\dot\beta_i\alpha_i e^{-|x_i|} \sgn(x_i) \big) 
\end{equation}
and
\begin{equation}\label{utx_expr}
u_{tx} = -\sum_i \big( \dot\alpha_i e^{-|x_i|} \sgn(x_i) +\dot\beta_i\alpha_i(e^{-|x_i|}-2\delta(x_i)) \big)  . 
\end{equation}

To begin, 
we substitute the general $N$-peakon expression \eqref{multi-peakon}
into the weak equation \eqref{pde_weak}.
There are two ways we can then proceed.
One way which is commonly used is to 
assume $\beta_1 < \beta_2 <\cdots< \beta_N$ at a fixed $t>0$,
split up the integral over $x$ into corresponding intervals,
and integrate by parts,
similarly to the derivation of the single peakon solution.
Another way, which is considerably simpler, 
is to employ the following result which can be easily proven 
from distribution theory \cite{Gel-Shi}.

Let $h(x)$ be a distribution whose singular support 
is a set of a finite number of points $x=x_i$ in $\mathbb{R}$, 
and define its non-singular part 
\begin{equation}\label{fmean}
\left< h(x)  \right>= 
\begin{cases}
h (x), 
& x \neq x_i\\
\tfrac{1}{2}(h(x_i^+) + h(x_i^-)), 
& x = x_i
\end{cases}
\end{equation}
and its jump discontinuities 
\begin{equation}\label{fjump}
\left[ h \right]_{x_i} = h(x_i^+) - h(x_i^-) = \left[\left< h(x)\right> \right]_{x_i} . 
\end{equation}

\begin{lem}
(i) 
If $h(x)$ is a piecewise smooth, bounded function,
then $h'(x)$ is a distribution whose singular support is the set of points $x=\beta_i$
at which $h(x)$ is discontinuous,
and $\left< h'(x)\right>$ is a piecewise smooth, bounded function. 
The functions $\left< h'(x)\right>$ and $h(x)$ are related by the integral identity 
\begin{equation}\label{ibp_rel}
\int_{-\infty}^{\infty} \psi' h \, dx + \int_{-\infty}^{\infty} \psi \left< h' \right> \, dx
= - \sum_i (\psi \left[ h \right])|_{x=\beta_i}
\end{equation}
holding for any test function $\psi(x)$. 
(ii)
If $h_1(x)$ and $h_2(x)$ are piecewise smooth, bounded functions,
then 
\begin{equation}\label{prod_rels}
\int_{-\infty}^{\infty} \psi \left< h_1 h_2 \right> \, dx
=\int_{-\infty}^{\infty} \psi h_1h_2 \, dx ,
\quad
\int_{-\infty}^{\infty} \psi \left< h_1 h_2' \right> \, dx
= \int_{-\infty}^{\infty} \psi h_1 \left< h_2' \right> dx
\end{equation}
are identities. 
\end{lem}

We will now use the identities \eqref{ibp_rel}--\eqref{prod_rels}
to combine and evaluate all of the terms in the weak equation \eqref{pde_weak}.

We start with the term
\begin{equation}\label{term1}
\begin{aligned}
-\iint_{-\infty}^{\infty} \psi_{xx} u_t \, dx \, dt 
= & 
- \iint_{-\infty}^{\infty} \psi  \left<\left< u_{tx} \right>_x\right> \, dx  \, dt 
\\&\qquad
+\sum_i \int_{-\infty}^{\infty}\big( 
\psi_x(t,\beta_i) \left[u_t\right]_{\beta_i} - \psi(t,\beta_i) \left[\left< u_{tx} \right>\right]_{\beta_i}
\big) \, dt . 
\end{aligned}
\end{equation}
By using expression \eqref{utx_expr},
we see 
\begin{equation}\label{utx_u_rel}
\left<\left< u_{tx} \right>_x\right>  = u_t , 
\end{equation}
so then the term \eqref{term1} can be combined with 
the other term in equation \eqref{pde_weak} involving $u_t$. 
These two terms yield 
\begin{equation}\label{ut_terms}
\iint_{-\infty}^{\infty}
(\psi u_t - \psi_{xx} u_t) \, dx \, dt =
\sum_i \int_{-\infty}^{\infty} \big( 
\psi_x(t,\beta_i) \left[u_t\right]_{\beta_i} -\psi(t,\beta_i) \left[\left< u_{tx} \right>\right]_{\beta_i} 
\big) \, dt . 
\end{equation}

We next consider the term 
\begin{equation}
\iint_{-\infty}^{\infty} \psi_{xx} G  \, dx \, dt 
=  -\iint_{-\infty}^{\infty} \psi_x \left<D_x G \right> \, dx  \, dt 
- \sum_i \int_{-\infty}^{\infty} \psi_x(t,\beta_i) \left[G\right]_{\beta_i} \, dt . 
\end{equation}
Expanding the total derivative $D_x G$, 
we note
\begin{equation}
\begin{aligned}
\iint_{-\infty}^{\infty} \psi_x \left<D_x G \right> \, dx  \, dt 
& = \iint_{-\infty}^{\infty} \psi_x ( \left<  G_u\right>\left< u_x \right>+ \left<G_{u_x}\left>\right< u_{xx} \right> )\, dx  \, dt 
\\
& = \iint_{-\infty}^{\infty} \psi_x ( G_u u_x + g u )\, dx  \, dt 
\end{aligned}
\end{equation}
by using 
\begin{equation}\label{uxx_u_rel}
\left< u_{xx} \right>  = u 
\end{equation}
which follows from expression \eqref{uxx_expr}. 
Thus, we have 
\begin{equation}
\iint_{-\infty}^{\infty} \psi_{xx} G  \, dx \, dt 
= -\iint_{-\infty}^{\infty} \psi_x ( G_u u_x + g u )\, dx  \, dt 
- \sum_i \int_{-\infty}^{\infty} \psi_x(t,\beta_i) \left[G\right]_{\beta_i} \, dt
\end{equation}
which then can be combined with the similar term in equation \eqref{pde_weak}, 
yielding
\begin{equation}\label{G_terms}
-\iint_{-\infty}^{\infty}( \psi_x (g u + G_u u_x) + \psi_{xx} G  )  \, dx \, dt 
= \sum_i \int_{-\infty}^{\infty} \psi_x(t,\beta_i) \left[G\right]_{\beta_i} \, dt . 
\end{equation}

Last, we consider the remaining terms in equation \eqref{pde_weak}. 
These terms combine similarly to the previous two terms, 
which gives 
\begin{equation}\label{F_terms}
\iint_{-\infty}^{\infty}( \psi ( f u + F_u u_x) + \psi_x F)  \, dx \, dt 
=- \sum_i \int_{-\infty}^{\infty} \psi(t,\beta_i) \left[F\right]_{\beta_i}\, dt . 
\end{equation}

Now, by combining all of the terms \eqref{ut_terms}, \eqref{G_terms}, \eqref{F_terms},
we obtain
\begin{equation}\label{all_terms}
0 = 
\sum_i \int_{-\infty}^{\infty} \psi_x(t,\beta_i) (\left[u_t\right]_{\beta_i} +\left[G\right]_{\beta_i}) \, dt 
-\sum_i \int_{-\infty}^{\infty} \psi(t,\beta_i) (\left[F\right]_{\beta_i} +\left[\left<u_{tx}\right>\right]_{\beta_i}) \, dt . 
\end{equation}
This equation \eqref{all_terms} will hold for all test functions $\psi$ 
if and only if 
\begin{equation}\label{conds}
\left[u_t\right]_{\beta_i} +\left[G\right]_{\beta_i} =0,
\quad
\left[\left<u_{tx}\right>\right]_{\beta_i} +\left[F\right]_{\beta_i} =0 . 
\end{equation}
The jump terms involving $t$-derivatives of $u$ are given by
\begin{equation}\label{ut_jumps}
\left[u_t\right]_{\beta_i} = 2\alpha_i \dot\beta_i, 
\quad
\left[\left<u_{tx}\right>\right]_{\beta_i} = - 2\dot\alpha_i
\end{equation}
which can be derived directly from expressions \eqref{ut_expr} and \eqref{utx_expr}.
For the other jump terms, 
we first note 
\begin{gather}\label{u_ux_jumps}
u(t,\beta_i^\pm) = u(t,\beta_i) 
= \sum_j \alpha_j e^{-|\beta_i-\beta_j|}
= \alpha_i+\sum_{j\neq i} \alpha_j e^{-|\beta_i-\beta_j|} , 
\\
u_x(t,\beta_i)  = -\sum_{j} \sgn(\beta_i-\beta_j) \alpha_j e^{-|\beta_i-\beta_j|}
= -\sum_{j\neq i} \sgn(\beta_i-\beta_j) \alpha_j e^{-|\beta_i-\beta_j|} , 
\\
u_x(t,\beta_i^\pm) = \mp\alpha_i +  u_x(t,\beta_i)
= \mp\alpha_i -\sum_{j\neq i} \sgn(\beta_i-\beta_j) \alpha_j e^{-|\beta_i-\beta_j|} 
\end{gather}
using expressions \eqref{u_expr} and \eqref{ux_expr},
where we now extend the definition \eqref{sgn} of the sign function 
by defining
\begin{equation}
\sgn(0)=0. 
\end{equation}
Then we have 
\begin{equation}\label{FG_jumps} 
\begin{aligned}
& \left[F\right]_{\beta_i} 
= F(u(t,\beta_i),-\alpha_i+u_x(t,\beta_i)) - F(u(t,\beta_i),\alpha_i+u_x(t,\beta_i)), 
\\
& \left[G\right]_{\beta_i} 
= G(u(t,\beta_i),-\alpha_i+u_x(t,\beta_i)) - G(u(t,\beta_i),\alpha_i+u_x(t,\beta_i)), 
\end{aligned}
\end{equation}
which can be expressed directly in terms of $f$ and $g$ 
through the integrals \eqref{F} and \eqref{G}. 

Consequently, 
after expressions \eqref{ut_jumps}, \eqref{u_ux_jumps}, \eqref{FG_jumps}
are substituted into equation \eqref{conds}, 
we have established the following main result. 

\begin{thm}\label{thm:multi_peakon}
Every equation \eqref{fg-fam} in the $fg$-family
possesses $N$-peakon weak solutions \eqref{multi-peakon}
for arbitrary $N\geq 1$ 
in which the time-dependent amplitudes $\alpha_i(t)$ and positions $\beta_i(t)$ 
($i=1,\ldots,N$)
satisfy the system of ODEs 
\begin{align}
\dot\alpha_i & =\tfrac{1}{2}( F(U_i,V_i-\alpha_i) - F(U_i,V_i+\alpha_i) ) , 
\label{ampl_eqn}
\\
\dot\beta_i & = \tfrac{1}{2}( G(U_i,V_i+\alpha_i) - G(U_i,V_i-\alpha_i) )/\alpha_i , 
\label{pos_eqn}
\end{align}
with $F$ and $G$ given by expressions \eqref{F} and \eqref{G}, 
and with 
\begin{equation}\label{UV}
U_i = \sum_{j=1}^{N} \alpha_j e^{-|\beta_{i,j}|} = u(t,\beta_i), 
\quad
V_{i}= -\sum_{j=1}^{N} \sgn(\beta_{i,j}) \alpha_j e^{-|\beta_{i,j}|} = u_x(t,\beta_i), 
\quad
\beta_{i,j} = \beta_i-\beta_j . 
\end{equation}
\end{thm}

It is straightforward to show that the general ODE system \eqref{ampl_eqn}--\eqref{UV}
reproduces the multi-peakon dynamical systems found in the literature 
\cite{Cam-Hol,Cam-Hol-Hym,Deg-Hol-Hon,Wan-Hon,Gui-Liu-Olv-Qu}
using weak formulations for all of the peakon equations in Table~\ref{table:peakon_eqns}.

This theorem is very interesting because it shows that 
all equations in the $fg$-family \eqref{fg-fam} admit $N$-peakons for arbitrary $N\geq 1$. 
Consequently, 
the existence of multi-peakons for any given equation \eqref{fg-fam} 
does not rely on the equation explicitly exhibiting any integrability properties
or any Hamiltonian structure.
Indeed, according to the integrability classification in Ref.~\cite{Nov},
there are many non-integrable quadratic and cubic peakon equations belonging to this family \eqref{fg-fam}, 
and no Hamiltonian structure is known for these non-integrable equations
other than the $b$-family equation \eqref{b-fam}. 

As another consequence, 
all equations in $fg$-family \eqref{fg-fam} admit a $1$-peakon solution, 
possibly having a time-dependent amplitude and a time-dependent speed, 
given by the case $N=1$. 

\begin{cor}
Any equation \eqref{fg-fam} in the $fg$-family 
possesses a generalized peakon solution of the form 
\begin{equation}\label{general-peakon}
u(t,x)=\alpha(t) \exp(-|x-\beta(t)|)
\end{equation}
whose amplitude $\alpha(t)$ and position $\beta(t)$
are given by the ODEs
\begin{equation}
\dot\alpha =\tfrac{1}{2}(F(\alpha,-\alpha)-F(\alpha,\alpha)),
\quad
\dot\beta = \tfrac{1}{2}(G(\alpha,\alpha)-G(\alpha,-\alpha))/\alpha ,
\end{equation}
where $F$ and $G$ are given by expressions \eqref{F} and \eqref{G}. 
This solution reduces to a travelling wave
\begin{equation}
u(t,x)=a \exp(-|x-ct|)
\end{equation}
with 
\begin{equation}\label{1peakon_cond}
\alpha = a = \const, 
\quad 
\dot\beta= c = \const\neq 0
\end{equation}
if and only if $f$ and $g$ have the forms \eqref{fg_1peakon}--\eqref{fg_1peakon'}. 
\end{cor}

Finally, 
the derivation of the $N$-peakon solutions indicates that if 
we try to allow the functions $f$ or $g$ in the $fg$-family \eqref{fg-fam} 
to depend on second-order (or higher) derivatives of $u$
then this will unavoidably lead to the resulting weak equation 
having products of Dirac delta distributions. 
In such a situation, there would seem to be no sensible way to show that 
a $N$-peakon expression is a weak solution. 
Therefore, 
the $fg$-family is arguably the most general family of nonlinear dispersive wave equations of the form \eqref{gen-fam}
that can possess multi-peakon weak solutions \eqref{multi-peakon}.

\section{Hamiltonian structure}\label{sec:hamiltonian-struct}

Each of the known integrable peakon equations 
--- CH \eqref{CH}, DP \eqref{DP}, mCH \eqref{FORQ}, N \eqref{N} ---
has a bi-Hamiltonian structure
\begin{equation}\label{bi-Ham-eqn}
m_t =\Hop (\delta H/\delta m) =\Eop (\delta E/\delta m) 
\end{equation}
where $\Hop$ and $\Eop$ are a pair of compatible Hamiltonian operators,
and $H$ and $E$ are corresponding Hamiltonian functionals. 
This general structure can be expressed equivalently in terms of 
a pair of associated Poisson brackets
\begin{equation}\label{PB-eqn}
m_t =\{ m, H\}_\Hop =\{ m, E\}_\Eop 
\end{equation}
defined by 
\begin{equation}\label{PB}
\{ H_1, H_2\}_\Dop = \int_{-\infty}^{\infty} (\delta H_1/\delta m)\Dop(\delta H_2/\delta m)\, dx
\end{equation}
on the real line. 
In particular, 
an operator $\Dop$ is a Hamiltonian operator when and only when 
the bracket \eqref{PB} is a Poisson bracket 
(namely, it is skew and satisfies the Jacobi identity). 

As is well-known, it is interesting that
both the CH equation \eqref{CH} and the mCH equation \eqref{FORQ}
share one Hamiltonian structure: 
\begin{align}
& m_t = -u_x m - (um)_x = \Hop(\delta H_{\text{CH}}/\delta m)
\label{CH_hamil_struc}\\
& m_t =  - ((u^2-u_x^2)m)_x = \Hop(\delta H_{\text{mCH}}/\delta m)
\label{FORQ_hamil_struc}
\end{align}
with 
\begin{equation}\label{Hop}
\Hop =D_x^3 -D_x =-\Delta D_x = -D_x \Delta,
\quad
\Delta = 1-D_x^2
\end{equation}
and
\begin{align}
H_{\text{CH}} & 
= \int_{-\infty}^{\infty} \tfrac{1}{2} u(u^2 + u_x^2)\, dx
= \int_{-\infty}^{\infty} \tfrac{1}{2} (u^2 - u_x^2) m \, dx
\label{CH_hamil}\\
H_{\text{mCH}} & 
= \int_{-\infty}^{\infty} ( \tfrac{1}{4} (u^2+u_x^2)^2 -\tfrac{1}{3}u_x^4 )\, dx
= \int_{-\infty}^{\infty} \tfrac{1}{4} u(u^2-u_x^2) m \, dx . 
\label{FORQ_hamil}
\end{align}
Note, in this Hamiltonian structure \eqref{CH_hamil_struc}--\eqref{FORQ_hamil_struc}, 
$u$ is viewed as a potential for $m$ through the relation 
\begin{equation}
m=\Delta u . 
\end{equation}
The variational derivatives with respect to $m$ 
then can be formulated in terms of variational derivatives with respect to $u$ 
by the identity 
\begin{equation}
\Delta\frac{\delta}{\delta m} = \frac{\delta}{\delta u} . 
\end{equation}
This yields 
\begin{align}
& m_t = -u_x m - (um)_x = -D_x(\delta H_{\text{CH}}/\delta u)
\label{CH_hamil_u}\\
& m_t =  - ((u^2-u_x^2)m)_x = -D_x(\delta H_{\text{mCH}}/\delta u)
\label{FORQ_hamil_u}
\end{align}
where $-D_x$ is a Hamiltonian operator with respect to $u$, 
which provides a useful alternative form for the common Hamiltonian structure of
the CH and mCH equations. 

Motivated by these Hamiltonian formulations \eqref{CH_hamil_u}--\eqref{FORQ_hamil_u},
we will now seek the most general form for the functions 
$f(u,u_x)$ and $g(u,u_x)$ in the $fg$-family \eqref{fg-fam}
so that 
\begin{equation}\label{Hamil_struc}
m_t = -f(u,u_x)m -(g(u,u_x)m)_x = -D_x (\delta H/\delta u) = -\Hop (\delta H/\delta m) 
\end{equation}
possesses a Hamiltonian structure using the Hamiltonian operator \eqref{Hop} 
common to the CH and mCH equations. 
There are two conditions for this structure \eqref{Hamil_struc} to exist, 
as determined by the necessary and sufficient relation 
$f(u,u_x)m +(g(u,u_x)m)_x = D_x (\delta H/\delta u)$. 
The first condition is that $f m = D_x A$ must hold 
for some differential function $A$ of $u$ and $x$-derivatives of $u$. 
Then the second condition is that $A + gm = \delta H/\delta u$ must hold 
for some functional $H= \int_{-\infty}^{\infty} B\, dx$
where the density $B$ is a differential function of $u$ and $x$-derivatives of $u$. 
We can formulate these two conditions as determining equations on $f$ and $g$
by using some tools from variational calculus \cite{Olv,Anc-review}, 
in particular, the Euler operator and the Helmholtz operator. 
This leads to two overdetermined linear systems of equations, 
which we can straightforwardly solve to find $f(u,u_x)$ and $g(u,u_x)$ explicitly. 
Details are shown in the Appendix. 

We obtain 
\begin{align}
f(u,u_x) & = u_x f_1(u^2-u_x^2) + f_0 \frac{u}{u^2-u_x^2}, 
\label{f_sol}
\\
g(u,u_x) & = g_1(u^2-u_x^2) + u f_1(u^2-u_x^2) + f_0 \frac{u_x}{u^2-u_x^2} , 
\label{g_sol}
\end{align}
where $f_1$, $g_1$ are arbitrary functions of $u^2-u_x^2$, 
and $f_0$ is an arbitrary constant. 
To avoid the occurrence of singularities when $u^2=u_x^2$, 
we will hereafter put $f_0=0$ and take $f_1$, $g_1$ to be continuous functions. 
Then we have the following result.  

\begin{prop}\label{prop:fg-Hamil}
In the $fg$-family \eqref{fg-fam}, 
all non-singular peakon equations that share the CH-mCH Hamiltonian structure \eqref{Hamil_struc} 
are given by 
\begin{equation}\label{fg_Hamil_fam}
f(u,u_x) = u_x f_1(u^2-u_x^2),
\quad
g(u,u_x) = u f_1(u^2-u_x^2) + g_1(u^2-u_x^2)
\end{equation} 
where $f_1$ is an arbitrary $C^0$ function of $u^2-u_x^2$, 
and $g_1$ is an arbitrary $C^1$ function of $u^2-u_x^2$. 
The Hamiltonian structure takes the form 
\begin{equation}\label{fg_Hamil_struct}
m_t = -u_x f_1(u^2-u_x^2)  m - ((u f_1(u^2-u_x^2)  + g_1(u^2-u_x^2))m)_x
=-D_x\Delta (\delta H(u)/\delta m)
\end{equation} 
with 
\begin{equation}\label{fg_Hamil}
H(u) = \int_{-\infty}^{\infty} 
\tfrac{1}{2}\big( F_1(u^2-u_x^2) + u\tilde G_1(u^2-u_x^2) \big) m \, dx
\end{equation}
where 
\begin{equation}\label{fg_Hamil_FG}
\begin{aligned}
F_1 & =\int_0^{u^2-u_x^2} f_1(y)\,dy 
= (u^2-u_x^2) \int_0^1 f_1(\lambda(u^2-u_x^2))\,d\lambda, 
\\
\tilde G_1 & =(u^2-u_x^2)^{-1}\int_0^{u^2-u_x^2} g_1(y)\,dy 
= \int_0^1 g_1(\lambda(u^2-u_x^2))\,d\lambda . 
\end{aligned}
\end{equation}
\end{prop}

This Hamiltonian structure can be reformulated as a variational formulation. 
We write $u=v_x$ and use the variational relations 
$D_x\E_u = -\E_v$ and $m_t = -\tfrac{1}{2}\E_v(u_tu_x + v_t v_x)$
where $\E$ denotes the Euler operator (cf the Appendix), 
then we obtain 
\begin{equation}\label{varprinc_fg_fam}
0 = -\frac{\delta S(v)}{\delta v}
= m_t + u_x f_1(u^2-u_x^2)  m +((u f_1(u^2-u_x^2)  + g_1(u^2-u_x^2))m)_x
\end{equation}
given by the action principle
\begin{equation}\label{Lagr_fg_fam}
S(v) = \int_{0}^{\infty}\int_{-\infty}^{\infty} L \, dx \, dt ,
\quad
L = \tfrac{1}{2} \big( u_tu_x + v_t v_x + ( F_1(u^2-u_x^2) + u\tilde G_1(u^2-u_x^2) ) m \big) .
\end{equation}

The Hamiltonian family of multi-peakon equations \eqref{fg_Hamil_fam}--\eqref{fg_Hamil_FG}
reduces to the CH equation \eqref{CH_hamil_u}
when $f_1=1$, $g_1=0$,
and also reduces to the mCH equation \eqref{FORQ_hamil_u} 
when $f_1=0$, $g_1=u^2-u_x^2$. 
Thus, this family both unifies and generalizes 
these two peakon equations. 
We will call the subfamily with $g_1=0$ the \emph{CH-type family}, 
and likewise the subfamily with $f_1=0$ will be called the \emph{mCH-type family}. 

We will now discuss some of the general family's interesting features:
conservation laws for strong and weak solutions; 
single peakon and anti-peakon solutions;
and a minimizer principle.

\subsection{Conservation laws}

Conservation laws are important for analysis of the Cauchy problem 
as well as for the study of stability of peakon solutions. 
For the family of Hamiltonian multi-peakon equations 
\begin{equation}\label{fg_CH_mCH_fam}
m_t +u_x f_1(u^2-u_x^2)  m + ((u f_1(u^2-u_x^2)  + g_1(u^2-u_x^2))m)_x=0 , 
\end{equation}
we start by considering smooth solutions $u(t,x)$ on the real line. 

The Hamiltonian \eqref{fg_Hamil} of this family will yield a conserved energy 
under appropriate asymptotic decay conditions on $u(t,x)$. 
In particular, 
the local energy conservation law is given by the continuity equation
\begin{equation}
D_t E + D_x \Psi^E =0 
\end{equation}
where
\begin{equation}\label{ener}
E = \tfrac{1}{2}\big( F_1(u^2-u_x^2) + u\tilde G_1(u^2-u_x^2) \big) m 
\end{equation}
is the energy (Hamiltonian) density,
and where
\begin{equation}\label{ener_flux}
\Psi^E = \tfrac{1}{2}\big( 
u_{tx}(F_1 +u\tilde G_1) -u_t u_x\tilde G_1 +K^2 -K_x^2
\big)
\end{equation}
is the energy flux,
in terms of 
\begin{equation}
K=\Delta^{-1}\E_u(E) =\kappa * \E_u(E) . 
\end{equation}
The flux expression \eqref{ener_flux} is derived by 
first using the variational identity (see \eqref{frechet_euler_ident})
\begin{equation}
D_t E = u_t\E_u(E) + D_x(u_t (E_{u_x}-D_xE_{u_{xx}}) +u_{tx}E_{u_{xx}})
\end{equation}
with 
\begin{gather}
\E_u(E) = \tfrac{1}{2} F_1 + (u f_1 +g_1)m ,
\label{ener_grad}
\\
E_{u_x} = -u_x(F_1'+u\tilde G_1')m,
\quad
E_{u_{xx}} = -\tfrac{1}{2}(F_1+u\tilde G_1) ,
\label{ener_rels}
\end{gather}
and then applying integration by parts to the term $u_t(\delta E/\delta u)$ 
after substituting the Hamiltonian structure \eqref{Hamil_struc} 
expressed in the form 
$u_t = -\Delta^{-1}(D_x\E_u(E))$. 

From integration of the energy conservation law over $-\infty<x<\infty$, 
we see that the total energy 
$H(u)= \int_{-\infty}^{\infty} E\, dx$
will be conserved 
\begin{equation}\label{ener_conserved}
\frac{dH(u)}{dt} = - \Psi^E\Big|_{-\infty}^{\infty}
\end{equation}
for smooth solutions $u(t,x)$ that have vanishing asymptotic flux, 
$\Psi^E\rightarrow 0$ as $|x|\rightarrow \infty$. 

The Hamiltonian structure \eqref{Hamil_struc} itself has the form of 
a local conservation law 
\begin{equation}\label{mom_conslaw}
D_t m + D_x \Psi^M =0 
\end{equation}
for the momentum $m$, 
where the momentum flux is given by 
\begin{equation}\label{mom_flux}
\Psi^M = \E_u(E)
\end{equation}
in terms of the energy density \eqref{ener}. 
Consequently, the total momentum 
$M(u) = \int_{-\infty}^{\infty} m\, dx$ 
will be conserved 
\begin{equation}\label{mom_conserved}
\frac{dM(u)}{dt} = - \Psi^M\Big|_{-\infty}^{\infty}
\end{equation}
for smooth solutions $u(t,x)$ that have vanishing asymptotic flux, 
$\Psi^M\rightarrow 0$ as $|x|\rightarrow \infty$. 

Finally, we now show that the Hamiltonian structure also ensures conservation of 
the $H^1$ norm of solutions $u(t,x)$ with sufficient asymptotic decay. 

The time derivative of the $H^1$ density $u^2 + u_x^2$ is given by 
$D_t(u^2+u_x^2) = 2um_t +D_x(2uu_{tx}) = 2D_x(uu_{tx}-u\E_u(E)) + 2u_x\E_u(E)$
by using the Hamiltonian equation \eqref{Hamil_struc}. 
The term $2u_x\E_u(E)$ can be expressed as a total $x$-derivative 
\begin{equation}
u_x\E_u(E) = D_x (uF_1 +(u^2-u_x^2)\tilde G_1) 
\end{equation}
through the use of equations \eqref{ener_grad} and \eqref{fg_Hamil_FG} 
as well as the identity $2u_xm = D_x(u^2-u_x^2)$. 
This term $u_x(\delta E/\delta u)$ combines with the total $x$-derivative term 
$D_x(-2u\E_u(E))$ 
to give $D_x( (u^2-u_x^2)\tilde G_1 -2u(uf_1 +g_1)m )$. 
Thus, we obtain the local conservation law
\begin{equation}\label{H1_conslaw}
D_t (u^2+u_x^2) + D_x \Psi =0 
\end{equation}
where 
\begin{equation}\label{H1_flux}
\Psi = -2uu_{tx} -(u^2-u_x^2)\tilde G_1 +2u(uf_1 +g_1)m
\end{equation}
is the flux. 
This shows that 
$||u||^2_{H^1} =  \int_{-\infty}^{\infty}(u^2 + u_x^2)\, dx$
will be conserved 
\begin{equation}\label{H1_conserved}
\frac{d}{dt}||u||^2_{H^1} = - \Psi\Big|_{-\infty}^{\infty}
\end{equation}
for smooth solutions $u(t,x)$ that have vanishing asymptotic flux,
$\Psi\rightarrow 0$ as $|x|\rightarrow \infty$. 

A more relevant setting for analysis is given by reformulating
the family of Hamiltonian multi-peakon equations \eqref{fg_CH_mCH_fam} 
in the convolution form 
\begin{equation}\label{fg_CH_mCH_fam_strong_pde}
u_t = - {\kappa * ( u_x f_1(u^2-u_x^2)  m + ((u f_1(u^2-u_x^2)  + g_1(u^2-u_x^2))m)_x )}
= -{\kappa_x *\E_u(E)}
\end{equation}
using the inverse of the operator $1-\partial_x^2$
defined by the kernel \eqref{kernel}. 
Classical solutions of this equation \eqref{fg_CH_mCH_fam_strong_pde} on the real line
require $u$, $u_x$, $m$, and $u_t$ to be continuous functions 
having suitable asymptotic decay as $|x|\rightarrow \infty$. 
A strong solution will be a function $u(t,x)$ in an appropriate Sobolev space 
on $[0,\tau)\times\Rnum$, with some $\tau>0$,  
for which Sobolev embedding implies the requisite continuity and asymptotic decay: 
\begin{equation}\label{strong_u}
u(t,x)\in C^0([0,\tau);H^s(\Rnum)) \cap C^1([0,\tau);H^{s-2}(\Rnum)),
\quad
s>\tfrac{5}{2} . 
\end{equation}

It is straightforward to show that all of the preceding conservation laws 
hold for strong solutions. 
This involves first modifying the conserved densities by the addition of a trivial density 
that eliminates any terms containing $m$ and $u_{tx}$. 

In particular, 
the conservation laws \eqref{mom_conserved} for momentum 
and \eqref{ener_conserved} for energy 
can be expressed as
\begin{equation}\label{mom_integral}
\tilde M(u) = \int_{-\infty}^{\infty} u\, dx ,
\quad
\frac{d}{dt}\tilde M(u) = -\tilde \Psi^M\big|_{-\infty}^{\infty},
\quad
\tilde\Psi^M = K 
\end{equation}
and 
\begin{equation}\label{ener_integral}
\begin{gathered}
\tilde H(u) = \int_{-\infty}^{\infty} \big( \tfrac{1}{2}\big( uF_1 + (u^2-u_x^2)\tilde G_1 +2u_x(u\hat F_1 + \hat G_1)  \big)\,dx ,
\\
\frac{d}{dt}\tilde H(u) = -\tilde \Psi^E\big|_{-\infty}^{\infty},
\quad
\tilde\Psi^E = \tfrac{1}{2}(K^2 - K_x^2) + (u\hat F_1 + \hat G_1) K_x
\end{gathered}
\end{equation}
with 
\begin{equation}\label{hatF1hatG1}
\hat F_1 = \int_0^{u_x} f_1(u^2-y^2)\, dy, 
\quad
\hat G_1 = \int_0^{u_x} g_1(u^2-y^2)\, dy . 
\end{equation}
Similarly, 
the conservation law \eqref{H1_conserved} for the $H^1$ norm 
becomes 
\begin{equation}\label{H1_integral}
\frac{d}{dt}||u||^2_{H^1} = - \tilde\Psi\big|_{-\infty}^{\infty},
\quad
\tilde\Psi = uF_1 + (u^2-u_x^2)\tilde G_1 +u\tilde\Psi^M . 
\end{equation}
In all of these conservation laws, the fluxes will vanish since 
$u\rightarrow 0$ and $u_x\rightarrow 0$ for strong solutions as $|x|\rightarrow \infty$,
while $\kappa\rightarrow 0$ as $|x|\rightarrow \infty$. 
This establishes the following result. 

\begin{prop}\label{prop:fg-fam-conslaws}
In the Hamiltonian family \eqref{fg_CH_mCH_fam} of multi-peakon equations, 
strong solutions \eqref{strong_u} 
have conserved energy \eqref{ener_integral}, momentum \eqref{mom_integral}, 
and $H^1$ norm \eqref{H1_integral}. 
\end{prop}

For solutions with less regularity,
such as weak solutions, 
conservation of the energy and the $H^1$ norm will not hold in general. 
Specifically, 
weak solutions of the Hamiltonian family 
satisfy the integral equation 
\begin{equation}\label{Hamil_u_pde_weak}
0  = \int_{0}^{\infty}\int_{-\infty}^{\infty} \phi (u_t + {\kappa_x *\E_u(E)})\, dx \, dt 
\end{equation}
(cf Proposition~\ref{weak_conds})
for all test functions $\phi(t,x)$. 
Any conservation law holding for all weak solutions must arise directly from this integral equation,
by selecting a specific $\phi(t,x)$ with the right properties. 
It is clear that the global conservation law for the momentum \eqref{mom_integral}
can be obtained by choosing $\phi=1$ in a domain $[0,\tau)\times (-L,L)$ 
with $\phi\rightarrow 0$ rapidly outside this domain, 
and then taking $L\rightarrow \infty$. 
In contrast, 
the global conservation laws for the energy \eqref{ener_integral} 
and the $H^1$ norm \eqref{H1_integral} 
clearly cannot be extracted in this way. 

Finally, we remark that 
this weak formulation \eqref{Hamil_u_pde_weak} of the Hamiltonian family of multi-peakon equations \eqref{fg_CH_mCH_fam}
can be expressed as a variational principle. 
By first putting $\phi = \Delta\psi$ and $u=v_x$, 
we formally have 
$\int_{0}^{\infty}\int_{-\infty}^{\infty} \phi (u_t + {\kappa_x *\E_u(E)})\, dx \, dt 
= \int_{0}^{\infty}\int_{-\infty}^{\infty} \psi \E_v(\tilde L) \, dx \, dt$,
where $\tilde L = \tfrac{1}{2} (u_tu_x + v_t v_x) + \tilde E$,
with 
\begin{equation}\label{weak_ener_dens}
\tilde E = \tfrac{1}{2}\big( uF_1 + (u^2-u_x^2)\tilde G_1 +2u_x(u\hat F_1 + \hat G_1)  \big)
\end{equation}
given by the density of the energy integral \eqref{ener_integral}.
Next we use the variational relation 
$\int_{0}^{\infty}\int_{-\infty}^{\infty} \psi\E_v(\tilde L) \, dx \, dt = \int_{0}^{\infty}\int_{-\infty}^{\infty} \delta_\psi\tilde L  \, dx \, dt$
obtained via integration by parts,
where 
\begin{equation}
\delta_\psi \tilde L = 
\psi \parder{v}\tilde L + \psi_x \parder{v_x}\tilde L + \psi_{xx} \parder{v_{xx}}\tilde L +\psi_t \parder{v_t}\tilde L +\psi_{tx} \parder{v_{tx}}\tilde L
\end{equation}
is the Frechet derivative of $\tilde L$. 
Thus, the weak equation \eqref{Hamil_u_pde_weak} is equivalent to the variational principle 
\begin{equation}\label{weak_varprinc_fg_fam}
0  = \delta_\psi \tilde S(v),
\quad
\tilde S(v) = \int_{0}^{\infty}\int_{-\infty}^{\infty} \tilde L \, dx \, dt 
\end{equation}
with $\tilde S(v)$ being a modified action principle 
(analogous to the modified energy integral $\tilde H(u)$). 

This equation \eqref{weak_varprinc_fg_fam} gives a weak variational principle 
whose solutions $v(t,x)$ are potentials 
for weak solutions $u(t,x)$ of the multi-peakon equations \eqref{fg_CH_mCH_fam}.

\subsection{Single peakons}

The general Hamiltonian family \eqref{fg_CH_mCH_fam}
with $f_1$ and $g_1$ being arbitrary (non-singular) functions of $u^2-u_x^2$ 
possesses single peakon travelling-waves \eqref{singlepeakon},
as seen by applying the two existence conditions in Theorem~\ref{thm:single_peakon}. 

Specifically, we have 
$F= \int_{0}^{u_x} yf_1(u^2-y^2)\,dy = -\tfrac{1}{2}F_1(u^2-u_x^2)$
and $G= \int_{0}^{u_x} (uf_1(u^2-y^2)+g_1(u^2-y^2))\,dy 
= u_x(u\hat G_1(u,u_x) + \hat F_1(u,u_x))$
with $\hat G_1$ and $\hat F_1$ given by expressions \eqref{hatF1hatG1}. 
Since $F_1$, $\hat F_1$ and $\hat G_1$ are even functions of $u_x$, 
we see that condition \eqref{peakon_cond_F} on $F$ holds, 
while condition \eqref{peakon_cond_G} on $G$ becomes
$c=a\hat G_1(a,a) + \hat F_1(a,a)\neq 0$
where both $\hat G_1(a,a)$ and $\hat F_1(a,a)$ are even functions of $a$.
Hence, since $a\hat G_1(a,a)$ is an odd function of  $a$, 
this implies that condition \eqref{peakon_cond_G} is equivalent to having 
$\hat G_1(a,a)\neq 0$ or $\hat F_1(a,a)\neq 0$. 
From these two inequalities, we obtain the following result. 

\begin{prop}\label{prop:fg_Hamil_fam_singlepeakon}
The Hamiltonian family \eqref{fg_CH_mCH_fam} of multi-peakon equations 
possesses single peakon travelling-wave solutions 
$u_{(a)}= a e^{-|x-ct|}$ 
if $f_1$ and $g_1$ are (arbitrary) continuous functions. 
The speed-amplitude relation \eqref{fg_speed_ampl_rel} of these peakons is given by 
\begin{equation}\label{fg_Hamil_fam_speed_amplitude}
c = a c_1(a) + c_0(a)
\end{equation}
with 
\begin{equation}
c_1(a) = \int_0^1 f_1((1-\lambda^2)a^2) \,d\lambda ,
\quad
c_0(a) = \int_0^1 g_1((1-\lambda^2)a^2) \,d\lambda , 
\end{equation}
where the speed $c$ is non-zero if and only if $c_1\neq 0$ or $c_0\neq 0$. 
\end{prop}

The speed-amplitude relation \eqref{fg_Hamil_fam_speed_amplitude} 
is linear iff $f_1$ is a constant and $g_1$ is zero. 
As a result, in general, this relation is nonlinear,
where both $c_1(a)$ and $c_0(a)$ are non-constant even functions of $a$. 
This has some interesting consequences for the speed properties of 
peakons ($a=a_+>0$)  and anti-peakons ($a=a_-<0$). 

When peakons and anti-peakons with the same absolute amplitude $|a|=a_+=-a_-$ 
are considered, 
their respective speeds are
$c_+= |a| c_1(|a|) + c_0(|a|)$
and $c_- = -|a|c_1(|a|) +c_0(|a|)$,
whereby $c_+ - c_- = 2|a| c_1(|a|)$ is their speed difference. 
Therefore, in the case of the CH-subfamily $g_1=0$, 
peakons and anti-peakons have opposite speeds:
$c_+ = |a| c_1(|a|) = -c_-$. 
In contrast, in the case of the mCH-subfamily $f_1=0$, 
peakons and anti-peakons have equal speeds:
$c_+ = c_0(|a|) = c_-$. 

In general, 
we see that all peakons will move to the right, $c_+>0$, 
if $f_1(y)$  and $g_1(y)$ are non-negative functions for $y>0$,
since this implies both $c_0(a_+)$ and $c_1(a_+)$ are positive, 
so then $c_+=a_+c_1(a_+)+c_0(a_+)$ is a sum of two positive terms. 
In contrast, 
the direction of anti-peakons depends on the relative magnitudes of $f_1(y)$  and $g_1(y)$ for $y>0$, 
since if $f_1(y)$  and $g_1(y)$ are non-negative functions for $y>0$
then $c_- = c_0(|a_-|)-|a_-|c_1(|a_-|)$ is a difference of two positive terms. 

Because single (anti) peakons are travelling waves, 
their $H^1$ norm 
along with their momentum and energy 
will be trivially conserved. 
We obtain 
\begin{gather}
||u_{(a)}||_{H^1} = \sqrt{2}|a|,
\label{peakon_H1norm}
\\
\tilde M(u_{(a)}) = a,
\quad
\tilde H(u_{(a)}) = a^2\int_0^1 \big( g_1(a^2y)\, \arctanh(\sqrt{1-y}) +|a|f_1(a^2 y) \sqrt{1-y}\big)\,dy . 
\label{peakon_mom_ener}
\end{gather}

\subsection{Minimizer principle}

The peakon travelling waves obtained in Proposition~\ref{prop:fg_Hamil_fam_singlepeakon}
are solutions $u(t,x)=U(x-ct)$ of the weak travelling-wave equation \eqref{ode_weak}. 
This weak ODE is readily verified to have the form 
\begin{equation}\label{fg_fam_ode_weak}
0 = \int_{-\infty}^{\infty} \left( cU'(\psi-\psi'')  + \psi' \tilde E_{U} + \psi'' \tilde E_{U'} \right) \, d\xi
\end{equation}
where $\tilde E$ is the energy density \eqref{weak_ener_dens} 
evaluated for travelling waves. 
Here $\psi(\xi)$ is a test function, with $\xi=x-ct$. 
If we write $\phi=\psi'$, 
then in equation \eqref{fg_fam_ode_weak} 
the first term can be expressed as 
$\int_{-\infty}^{\infty} cU'(\psi-\psi'')\, d\xi 
= -\int_{-\infty}^{\infty} c(U\phi + U'\phi')\, d\xi
= -\tfrac{1}{2}c\delta_\phi ||U||^2_{H^1}$,
while the second and third terms are similarly given by 
$-\int_{-\infty}^{\infty} \big( \phi \tilde E_{U} + \phi' \tilde E_{U'} \big) \, d\xi
= -\delta_\phi \tilde H(U)$
where $\tilde H(U)$ is the energy integral \eqref{ener_integral} evaluated for travelling waves. 
Thus, 
the weak travelling-wave equation \eqref{ode_weak} 
has an equivalent formulation as a weak variational principle 
\begin{equation}
0= \delta_\phi \big( \tfrac{1}{2}c||U||^2_{H^1} + \tilde H(U) \big) 
\end{equation}
with $\phi(\xi)$ being an arbitrary test function. 

This is a counterpart of the variational principle \eqref{weak_varprinc_fg_fam} 
for weak solutions of the Hamiltonian family of multi-peakon equations \eqref{fg_CH_mCH_fam}. 
It can be used as a starting point to prove stability of peakon travelling waves. 
In particular, whenever the nonlinearities $f_1(u^2-u_x^2)$ and $g_1(u^2-u_x^2)$ 
are such that the energy integral $\tilde H(u)$ is positive definite, 
the functional $\tfrac{1}{2}c||U||^2_{H^1} + \tilde H(U)$ will also be positive definite.
If peakon travelling waves are the ground state (namely, minimizers) of this functional,
then they will be stable. 

A natural conjecture is that peakon travelling waves are the solutions of the minimizer principle
\begin{equation}
I := \inf_{U(\xi)} ||U||_{H^1} 
\quad\text{such that}\quad
\tilde H(U) + \tfrac{1}{2}c||U||^2_{H^1} = h(a) + c a^2 
\end{equation}
where $h(a)=\tilde H(u_{(a)})$ 
denotes the energy of peakon travelling waves \eqref{peakon_mom_ener}.

In the case of the CH equation \eqref{CH_hamil_u}
when $f_1=1$, $g_1=0$,
this minimizer principle reduces to the minimization problem \cite{ConMol-2001}
for Camassa-Holm peakons, 
for which peakon travelling waves are known to be unique solution (up to translations). 

A similar proof that peakon travelling waves are the unique solution of the minimization problem for the general Hamiltonian family \eqref{fg_CH_mCH_fam},
as well as a proof of stability, 
will be left for subsequent work.

\section{One-parameter generalizations of CH and mCH equations}\label{sec:examples}

The family of Hamiltonian peakon equations \eqref{fg_Hamil_fam} 
involves two arbitrary functions $f_1$ and $g_1$ of $u^2-u_x^2$. 
This family provides a wide generalization of 
both the CH equation \eqref{CH_hamil_u} and the mCH equation \eqref{FORQ_hamil_u}.
It can be viewed as an analog of the Hamiltonian family of generalized KdV equations 
\begin{equation}\label{KdV_fam}
u_t = f(u)u_x +u_{xxx} =D_x \Delta (\delta H/\delta u),
\quad
H = \int_{-\infty}^{\infty} ( F(u) - \tfrac{1}{2} u_x^2 )\, dx
\end{equation}
where $F'=f$. 
This generalized KdV family has a scaling invariant subfamily 
given by the gKdV equation \eqref{gKdV},
which involves a general nonlinearity power $p\neq 0$. 

We can obtain an analogous scaling invariant subfamily of 
Hamiltonian peakon equations \eqref{fg_Hamil_fam} 
by taking the functions $f_1$ and $g_1$ to have a power form:
$f_1(u^2-u_x^2) = a (u^2-u_x^2)^{p-1}$, 
$g_1(u^2-u_x^2) = b (u^2-u_x^2)^{q-1}$, 
where $a$, $b$, $p$, $q$ are constants. 
This yields 
$m_t + a u_x(u^2-u_x^2)^{p-1} m + ( a u(u^2-u_x^2)^{p-1} m + b (u^2-u_x^2)^{q-1} m )_x =0$
which will be invariant under the group of scaling transformations 
$u\rightarrow \lambda u$ and $t\rightarrow \lambda^{-k-1} t$ 
($\lambda\neq 0$)
if the nonlinearity powers are related by $2(q-p)=1$,
with $k = 2p-2=2q-3$. 

Hence, we have a three-parameter family of scaling-invariant 
Hamiltonian peakon equations 
\begin{equation}\label{CH-FORQ-fam}
m_t + a u_x(u^2-u_x^2)^{k/2} m + ( a u(u^2-u_x^2)^{k/2} m + b (u^2-u_x^2)^{(k+1)/2} m )_x =0
\end{equation}
which has the Hamiltonian structure \eqref{fg_Hamil_fam}--\eqref{fg_Hamil},
where the Hamiltonian is given by 
\begin{equation}\label{CH-FORQ-fam_hamil}
H = \int_{-\infty}^{\infty} 
\left( \tfrac{a}{k+2} (u^2-u_x^2)^{1/2} + \tfrac{b}{k+3} u  \right) 
(u^2-u_x^2)^{(k+1)/2} m \, dx . 
\end{equation}
This family \eqref{CH-FORQ-fam}--\eqref{CH-FORQ-fam_hamil}
unifies the CH equation ($k=0$, $b=0$) and the mCH equation ($k=1$, $a=0$), 
up to a scaling of $u$. 
It represents a close analog of the gKdV equation \eqref{gKdV}.

\subsection{Generalized CH equation}

By putting $k=2p-2$, $b=0$, $a=1$ 
in the three-parameter family \eqref{CH-FORQ-fam}--\eqref{CH-FORQ-fam_hamil},
we obtain the one-parameter family of generalized CH equations \eqref{gCH}. 
This gCH family can be written equivalently as 
\begin{equation}\label{gCH'}
m_t + \big( \tfrac{1}{2p}(u^2-u_x^2)^{p} + u(u^2-u_x^2)^{p-1} m \big)_x = 0 . 
\end{equation} 
Like the relationship between 
the gKdV equation \eqref{gKdV} and the ordinary KdV equation, 
the gCH equation \eqref{gCH} reduces to the CH equation \eqref{CH}
when $p=1$ and retains one Hamiltonian structure of the CH equation, 
\begin{equation}\label{gCH_hamil}
m_t = -D_x\Delta(\delta H_{\text{gCH}}/\delta m), 
\quad
H_{\text{gCH}} = \int_{-\infty}^{\infty} \tfrac{1}{2p} (u^2 - u_x^2)^{p} m \, dx . 
\end{equation} 
For all $p\geq 1$, the gCH equation \eqref{gCH} possesses 
peakon travelling-wave solutions and multi-peakon solutions.
Its strong solutions have conserved 
momentum \eqref{mom_integral}, 
energy \eqref{ener_integral}, 
and $H^1$ norm \eqref{H1_integral}. 
For setting up the peakon equations, 
we note that $f=u_xf_1$ and $g= uf_1$ with $f_1 = (u^2-u_x^2)^{p-1}$. 

Single peakon and anti-peakon solutions $u=a e^{-|x-ct|}$ 
(with $a>0$ and $a<0$, respectively) 
are given by the speed-amplitude relation
\begin{equation}\label{gCH-speed}
c(a) = \gamma_p a^{2p-1},
\quad
\gamma_p = \frac{\sqrt{\pi}}{2}\frac{\Gamma(p)}{\Gamma(p+1/2)},
\end{equation}
as obtained from Proposition~\ref{prop:fg_Hamil_fam_singlepeakon}. 
We see that $c(a)$ is an odd function of $a$,
and as a result, 
peakons and anti-peakons have the same speed but move in opposite directions
(peakons to the right, and anti-peakons to the left). 

From Theorem~\ref{thm:multi_peakon}, 
multi peakon and anti-peakon solutions 
$u = \sum_{i=1}^{N} \alpha_i(t) e^{-|x-\beta_i(t)|}$ for $N\geq 2$
are described by the dynamical system 
\begin{equation}\label{gCH-multi-peakon}
\begin{aligned}
\dot\alpha_i & = \tfrac{1}{4p}( H^+(U_i,V_i)  -H^-(U_i,V_i) ), 
\\
\dot\beta_i & = \tfrac{1}{2}(U_i+\alpha_i)( (V_i+\alpha_i)\hat H^+(U_i,V_i) - (V_i-\alpha_i)\hat H^-(U_i,V_i) )/\alpha_i ,
\end{aligned}
\quad
i=1,2,\ldots,N, 
\end{equation} 
with 
\begin{equation}\label{Hnotation}
\begin{aligned}
H^\pm(U_i,V_i) &= (U_i^2 -(V_i \pm \alpha_i)^2)^p 
=\sum_{j=0}^{p} (-1)^j\tbinom{p}{j} U_i^{2(p-j)} (V_i\pm\alpha_i)^{2j} ,
\\
\hat H^\pm(U_i,V_i) & = \int_0^1 (U_i^2 -(V_i\pm\alpha_i)^2\lambda^2)^p\,d\lambda
=\sum_{j=0}^{p} \tfrac{(-1)^j}{2j+1}\tbinom{p}{j} U_i^{2(p-j)} (V_i\pm\alpha_i)^{2j} ,
\end{aligned}
\end{equation}
where $U_i$ and $V_i$ are given in terms of the dynamical variables $(\alpha_j(t),\beta_j(t))$ by expression \eqref{UV}. 
The interaction terms in this dynamical system for $p\geq 2$ are 
considerably more complicated than for $p=1$, 
and qualitatively new features turn out to occur,
which we will investigate in section~\ref{sec:interactions}.

\subsection{Generalized mCH equation}

We put $k=2p-1$, $a=0$, $b=1$ 
in the three-parameter family \eqref{CH-FORQ-fam}--\eqref{CH-FORQ-fam_hamil},
which yields a one-parameter family of generalized mCH equations \eqref{gmCH}. 
This gmCH family reduces to the ordinary mCH equation \eqref{FORQ}
when $p=1$, and retains one of its Hamiltonian structures
\begin{equation}\label{gFORQ_hamil}
m_t = -D_x\Delta(\delta H_{\text{gmCH}}/\delta m), 
\quad
H_{\text{gmCH}} = \int_{-\infty}^{\infty} \tfrac{1}{2(p+1)}u(u^2-u_x^2)^p m \, dx . 
\end{equation}
For all $p\geq 1$, this equation \eqref{gCH} possesses 
peakon travelling-wave solutions and multi-peakon solutions.
Its strong solutions have conserved 
momentum \eqref{mom_integral}, 
energy \eqref{ener_integral}, 
and $H^1$ norm \eqref{H1_integral}. 
To set up the peakon equations, 
we note that $f=0$ and $g=g_1= (u^2-u_x^2)^{p}$. 

Single peakon and anti-peakon solutions $u=a e^{-|x-ct|}$ 
(with $a>0$ and $a<0$, respectively) 
are given by the speed-amplitude relation
\begin{equation}\label{gmCH-speed}
c(a) = \gamma_{p+1} a^{2p}
\end{equation}
as obtained from Proposition~\ref{prop:fg_Hamil_fam_singlepeakon}
(where $\gamma_p$ is given in equation \eqref{gCH-speed}). 
%\gamma_p = \frac{\sqrt{\pi}}{2}\frac{\Gamma(p+1)}{\Gamma(p+3/2)}
We see that $c(a)$ is a positive, even function of $a$,
and hence both peakons and anti-peakons have the same speed and move to the right. 

From Theorem~\ref{thm:multi_peakon}, 
multi peakon and anti-peakon solutions 
$u = \sum_{i=1}^{N} \alpha_i(t) e^{-|x-\beta_i(t)|}$ for $N\geq 2$
are described by the dynamical system 
\begin{equation}\label{gmCH-multi-peakon}
\dot\alpha_i = 0 ,
\quad
\dot\beta_i = \tfrac{1}{2}( (V_i+\alpha_i)\hat H^+(U_i,V_i) - (V_i-\alpha_i)\hat H^-(U_i,V_i) )/\alpha_i ,
\quad
i=1,2,\ldots,N
\end{equation} 
using the notation \eqref{Hnotation},
with $U_i$ and $V_i$ being given by expressions \eqref{UV}
in terms of the dynamical variables $(\alpha_j(t),\beta_j(t))$. 
Similarly to the gCH equation, 
this dynamical system turns out to exhibit 
qualitatively new features for $p\geq 2$ 
compared to the ordinary mCH equation with $p=1$. 
We will investigate these features in section~\ref{sec:interactions}.

\subsection{Unification of generalized CH and mCH equations}

We can unify the gCH equation \eqref{gCH} and the gmCH equation \eqref{gmCH}
into a single one-parameter family by choosing the coefficients $a$ and $b$ 
to be suitable functions of $k$ 
in the three-parameter family \eqref{CH-FORQ-fam}--\eqref{CH-FORQ-fam_hamil}. 
Specifically, $a(k)$ and $b(k)$ need to satisfy 
$a(0)\neq 0$, $a(1)=0$, $b(0)=0$, $b(1)\neq 0$. 
For example, simple choices are $a(k)=1-k$ and $b(k)= k$. 
The resulting unified family
\begin{equation}
m_t + (1-k) u_x(u^2-u_x^2)^{k/2} m + ( ((1-k) u(u^2-u_x^2)^{k/2}  + k (u^2-u_x^2)^{(k+1)/2}) m )_x =0
\end{equation}
has the same properties as the gCH and gmCH equations.

\section{New behaviour in peakon interactions for higher nonlinearities}\label{sec:interactions}

The qualitative behaviour of (anti-) peakon solutions of 
the gCH equation \eqref{gCH} and the gmCH equation \eqref{gmCH}
when the nonlinearity power is $p=1,2$
will now be investigated. 
There turns out to be a significant difference in how the (anti-) peakons interact 
in the case $p=2$ compared to the case $p=1$. 

To proceed, we will consider the $N=2$ peakon solutions of both equations, 
by using the respective dynamical systems \eqref{gCH-multi-peakon} and \eqref{gmCH-multi-peakon}. 
Specifically, 
we will look at how the separation between two (anti-) peakons 
behaves with time $t$.

\subsection{2-peakon solutions of the gCH equation}

The $p=1$ case of the gCH equation \eqref{gCH} is the ordinary CH equation. 
$N=2$ peakon solutions \eqref{multi-peakon} 
are given by the dynamical system 
\begin{equation}
\dot\alpha_1 = -\dot\alpha_2 = \sgn(\beta_{1,2}) \alpha_1\alpha_2 e^{-|\beta_{1,2}|},
\quad
\dot\beta_1 = \alpha_1 + \alpha_2 e^{-|\beta_{1,2}|},
\quad
\dot\beta_2 = \alpha_2 + \alpha_1 e^{-|\beta_{1,2}|}
\end{equation}
where $\beta_{1,2} = \beta_1-\beta_2$ is the separation between 
the two (anti-) peakons or the peakon and the anti-peakon. 
This system has two constants of motion
$M=\alpha_1+\alpha_2$ 
and $E= \tfrac{1}{2}( \alpha_{1,2}^2(1-e^{-|\beta_{1,2}|}) + M^2e^{-|\beta_{1,2}|} )>0$,
where $\alpha_{1,2} = \alpha_1 -\alpha_2$ is the amplitude difference between the two (anti-) peakons. 
These quantities arise from the conserved momentum \eqref{mom_integral} and energy \eqref{ener_integral} evaluated for two (anti-) peakons. 
Using the constants of motion, 
we obtain 
\begin{equation}\label{CH-p=1-amplitude}
\alpha_{1,2} = \pm \sqrt{2E-M^2e^{-|\beta_{1,2}|}}/\sqrt{1- e^{-|\beta_{1,2}|}}
\end{equation}
and 
\begin{equation}\label{CH-p=1-separation}
%\dot\alpha_{1,2} = \tfrac{1}{2}\sgn(\beta_{1,2})  (E-\alpha_{1,2}^2) ,
\dot\beta_{1,2} =\pm\sgn(\beta_{1,2})  \sqrt{1- e^{-|\beta_{1,2}|}}\sqrt{2E-M^2 e^{-|\beta_{1,2}|}} . 
\end{equation}
This ODE can be straightforwardly integrated. 
The qualitative behaviour of $\beta_{1,2}(t)$, however, can be more easily found by 
looking at the collision points, defined by $\beta_{1,2}=0$, 
and the turning points, defined by $\dot\beta_{1,2}=0$. 
There are essentially two different types of behaviour, 
depending on the constant of motion 
\begin{equation}
\mu = \frac{2E}{M^2} -1 . 
\end{equation}
Note, since $\mu M^2 = -4\alpha_1\alpha_2(1-e^{-|\beta_{1,2}|})$, 
then $\mu>0$ implies $\alpha_1\alpha_2<0$, 
representing a peakon and an anti-peakon, 
and $\mu<0$ implies $\alpha_1\alpha_2>0$, 
representing two peakons or two anti-peakons. 
(The case $\mu=0$ corresponds to $\alpha_1=0$ or $\alpha_2=0$, which is trivial.)

When $\mu >0$, 
the only turning point is $\beta_{1,2}=0$ which is also a collision point. 
The constant of motion $\mu = ((\alpha_{1,2}/M)^2-1)(1-e^{-|\beta_{1,2}|})$
then shows that $\beta_{1,2}\to 0$ iff $|\alpha_{1,2}|\to \infty$.
Moreover, from the dynamical ODE \eqref{CH-p=1-separation}, 
we see that $\beta_{1,2}(t)$ will reach zero at a finite time. 
This means that a peakon and an anti-peakon will collide 
such that their relative amplitude blows up. 
See Figure~\ref{fig:gCH-p=1-muispos}. 

A different behaviour occurs when $\mu<0$.
There is an turning point $|\beta_{1,2}|=-\ln(1+\mu)$, 
at which $\alpha_{1,2}=0$. 
The dynamical ODE \eqref{CH-p=1-separation} shows that 
this turning point will be reached at a finite time,
and $|\beta_{1,2}|$ will be increasing before and after this time. 
Consequently, 
this represents two (anti-) peakons that are approaching each other, 
reach a minimum separation given by $|\beta_{1,2}|=-\ln(1+\mu)$
where their amplitudes are equal,
and move away from each other. 
In particular, their separation goes to infinity in the asymptotic past and future. 
This behaviour describes an elastic ``bounce'' interaction. 
See Figure~\ref{fig:gCH-p=1-muisneg}. 

\begin{figure}[H]
\centering
\begin{subfigure}[t]{.45\textwidth}
\includegraphics[width=\textwidth]{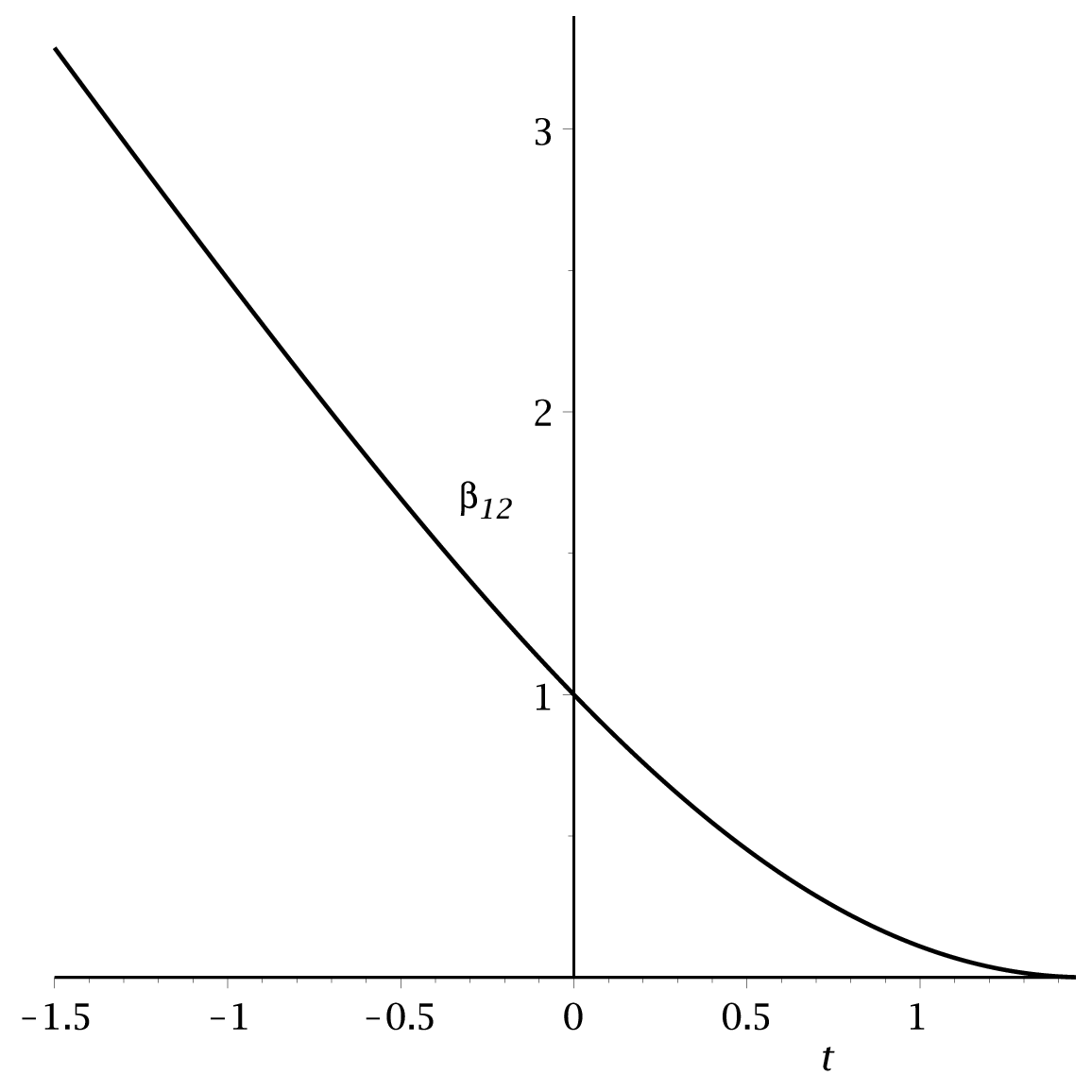}
\captionof{figure}{Separation}
%\label{fig:gCH-p=1-muispos-position}
\end{subfigure}%
\begin{subfigure}[t]{.45\textwidth}
\includegraphics[width=\textwidth]{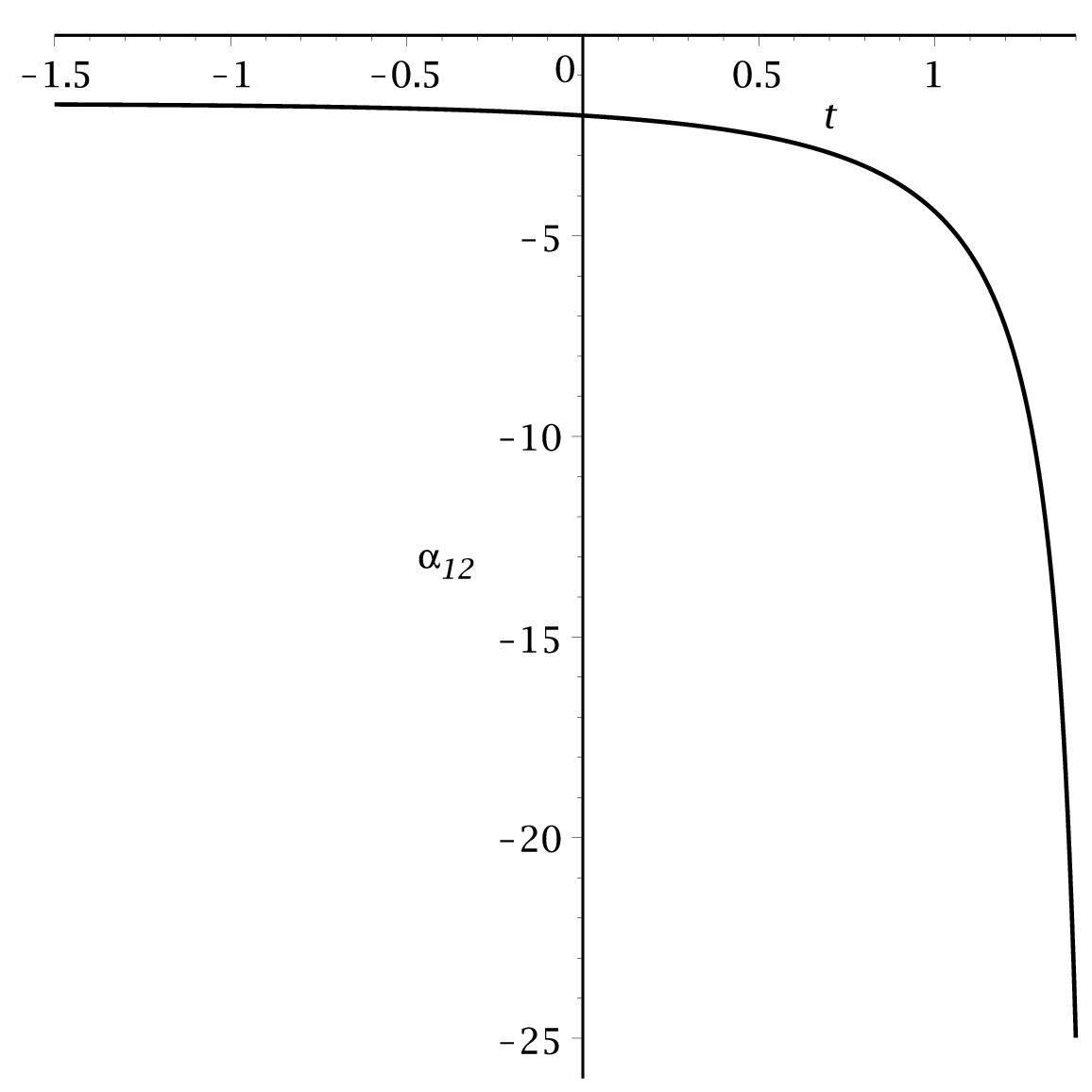}
\captionof{figure}{Relative amplitude}
%\label{fig:gCH-p=1-muispos-amplitude}
\end{subfigure}
\caption{Relative position and amplitude for CH 2-peakon collision ($\mu>0$)}
\label{fig:gCH-p=1-muispos}
\end{figure}

\begin{figure}[H]
\centering
\begin{subfigure}[t]{.45\textwidth}
\includegraphics[width=\textwidth]{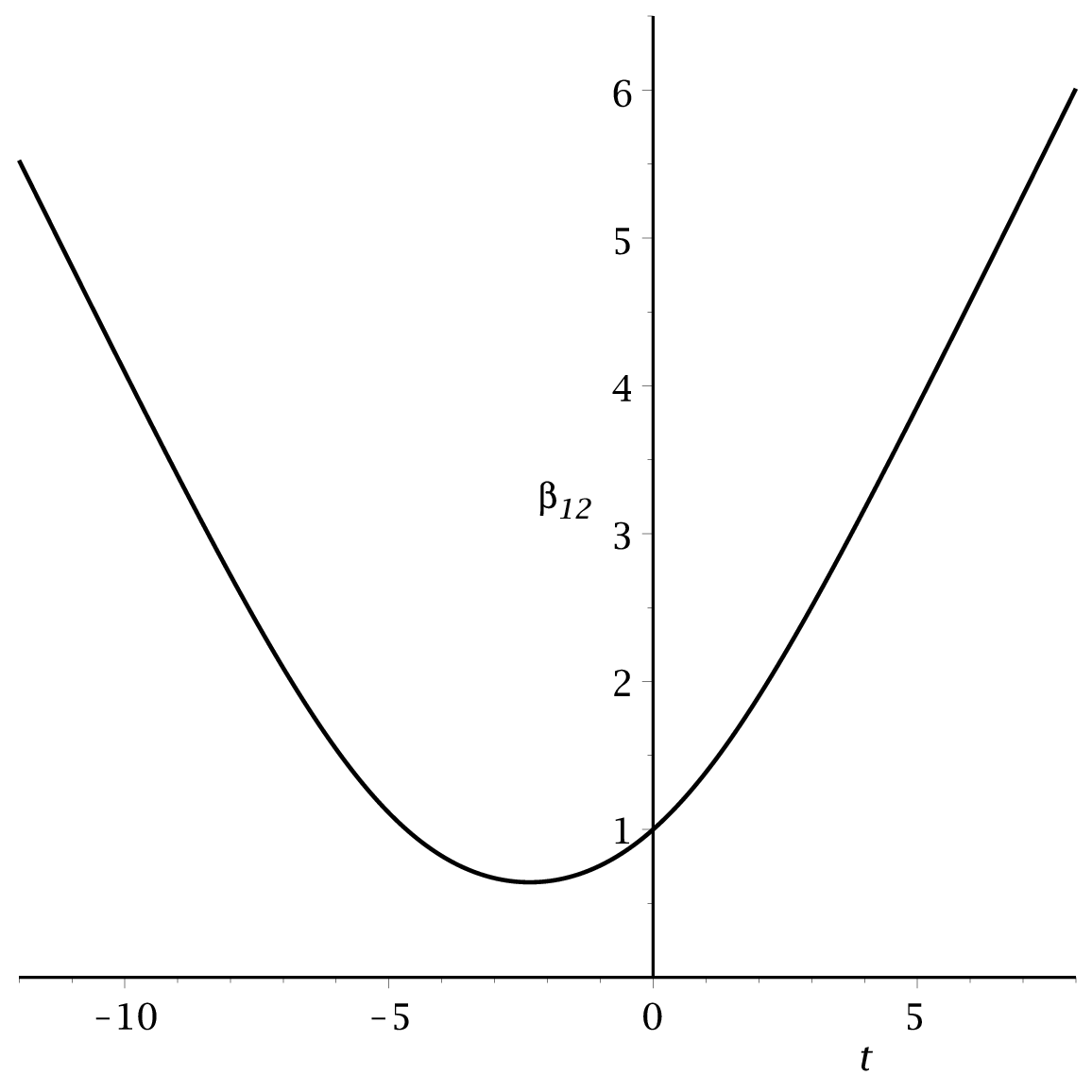}
\captionof{figure}{Separation}
%\label{fig:gCH-p=1-muisneg-position}
\end{subfigure}%
\begin{subfigure}[t]{.45\textwidth}
\includegraphics[width=\textwidth]{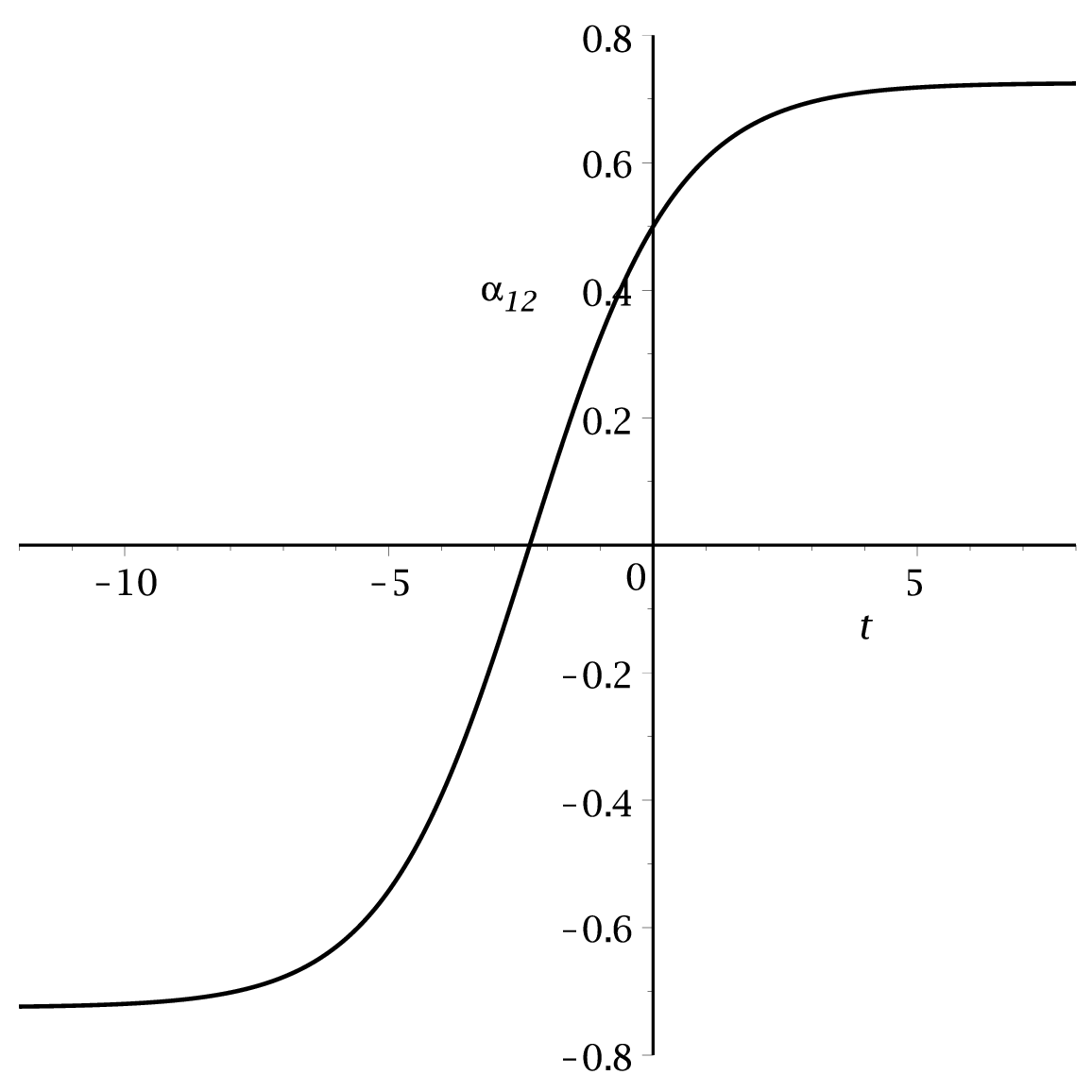}
\captionof{figure}{Relative amplitude}
%\label{fig:gCH-p=1-muisneg-amplitude}
\end{subfigure}
\caption{Relative position and amplitude for CH 2-peakon elastic interaction ($\mu<0$)}
\label{fig:gCH-p=1-muisneg}
\end{figure}

Similar dynamical behaviour is known to occur for all $N\geq 2$, 
using the integrability properties of the CH equation 
\cite{CheLiuPen,BeaSatSzm}. 

Next we consider the case $p=2$, which describes a nonlinear generalization of the CH equation. 
The dynamical system \eqref{gCH-multi-peakon} for $N=2$ in this case is given by 
\begin{gather}
\dot\alpha_1 = -\dot\alpha_2 = 2\sgn(\beta_{1,2}) \alpha_1^2\alpha_2^2 e^{-2|\beta_{1,2}|},
\\
\dot\beta_1 = (\alpha_1 + \alpha_2 e^{-|\beta_{1,2}|})(\tfrac{2}{3}\alpha_1^2 + 2\alpha_1\alpha_2 e^{-|\beta_{1,2}|}), 
\quad
\dot\beta_2 = (\alpha_2 + \alpha_1 e^{-|\beta_{1,2}|})(\tfrac{2}{3}\alpha_2^2 + 2\alpha_1\alpha_2 e^{-|\beta_{1,2}|}) . 
\end{gather}
This system again has the momentum $M=\alpha_1+\alpha_2$ as a constant of motion,
however energy is no longer conserved. 
Instead, a second constant of motion appears when we look at the dynamical equations
for $\alpha_{1,2}$ and $\beta_{1,2}$:
\begin{align}
\dot \alpha_{1,2} &= \tfrac{1}{4}\sgn(\beta_{1,2}) (M^2-\alpha_{1,2}^2)^2e^{-2|\beta_{1,2}|} ,
\label{gCH-aeqn}
\\
\dot \beta_{1,2} & = \tfrac{1}{6} \alpha_{1,2}(3M^2 +\alpha_{1,2}^2 + \tfrac{1}{6}(M^2 -\alpha_{1,2}^2)(4-3 e^{-|\beta_{1,2}|})e^{-|\beta_{1,2}|}) . 
\label{gCH-beqn}
\end{align}
The transformation $e^{|\beta_{1,2}|}= 4(M^2-\alpha_{1,2}^2)y/(3M^2+\alpha_{1,2}^2)$
leads to a separable Abel equation for $y(\alpha_{1,2})$, 
which possesses an explicit first integral. 
The corresponding constant of motion in terms of $\beta_{1,2}$ and $\alpha_{1,2}$ is given by 
\begin{equation}\label{gCH-FI}
C= \frac{3M^2 + \alpha_{1,2}^2 +3(M^2 -\alpha_{1,2}^2)e^{-|\beta_{1,2}|}}{(3M^2 + \alpha_{1,2}^2 +(M^2 -\alpha_{1,2}^2)e^{-|\beta_{1,2}|})^3}.
\end{equation}
This equation represents a cubic polynomial in $\alpha_{1,2}^2$,
which thereby gives $\alpha_{1,2}$ as a function of $\beta_{1,2}$,
yielding a dynamical ODE for $\beta_{1,2}(t)$. 
Rather than work with the resulting ODE, 
we can determine the qualitative behaviour of $\beta_{1,2}(t)$
directly from the pair of ODEs \eqref{gCH-aeqn}--\eqref{gCH-beqn}, 
by looking at the collision points and turning points in terms of $\beta_{1,2}$. 
The various behaviours depend essentially on the constant of motion 
\begin{equation}\label{gCH-com}
\nu = 9M^4C .
\end{equation}
It is straightforward to show that $1>\nu>-\infty$. 
Hereafter, we will write $B_{1,2} = e^{-|\beta_{1,2}|}$ for convenience,
where $0<B_{1,2}\leq 1$. 
Note that $(\alpha_{1,2}/M)^2 >1$ implies $\alpha_1$ and $\alpha_2$ have opposite signs,
and that $(\alpha_{1,2}/M)^2 <1$ implies $\alpha_1$ and $\alpha_2$ have the same sign. 

From the ODE \eqref{gCH-beqn}, 
we find that the turning points $\dot\beta_{1,2}=0$ occur for 
$\alpha_{1,2}=0$ 
and 
$(\alpha_{1,2}/M)^2 = (3B_{1,2}^2 - 4B_{1,2} - 3)/((1-B_{1,2})(1-3B_{1,2}))\geq 13$.
In the first case, 
$B_{1,2}$ is the root of the cubic equation 
$\nu (B_{1,2}+3)^3 =27(B_{1,2}+1)$ 
obtained from combining equations \eqref{gCH-FI} and \eqref{gCH-com}, 
with $1>\nu\geq \tfrac{27}{32}$ due to $0<B_{1,2}\leq 1$. 
In the second case, 
since $3B_{1,2}^2 - 4B_{1,2} - 3 <0$ for $0<B_{1,2}\leq 1$,
we have the condition $1>B_{1,2}>\tfrac{1}{3}$ 
along with the cubic equation 
$48\nu B_{1,2}^2(1-B_{1,2}) = (1-3B_{1,2})^3$
again given by combining equations \eqref{gCH-FI} and \eqref{gCH-com}, 
where $\nu< 0$. 

We also find that the collision points $\beta_{1,2}=0$ occur for 
$(\alpha_{1,2}/M)^2 = (27 -32\nu)/9$,
which requires $\nu\leq\tfrac{27}{32}$. 
The only collision point that coincides with a turning point, 
$\beta_{1,2}=\dot\beta_{1,2}=0$,  
is given by $\alpha_{1,2}=0$ with $\nu=\tfrac{27}{32}$. 

When $\tfrac{27}{32}<\nu <1$, 
there is only a turning point, 
and from the dynamical ODEs \eqref{gCH-aeqn}--\eqref{gCH-beqn} 
we can show that this point is reached at a finite time,
with $\alpha_{1,2}=0$. 
This represents two (anti-) peakons whose separation 
reaches a minimum given by the root of $\nu (3+e^{-|\beta_{1,2}|})^3 =27(1+e^{-|\beta_{1,2}|})$,
where their amplitudes are equal,
and goes to infinity in the asymptotic past and future. 
Hence, the behaviour is an elastic ``bounce'' interaction. 
See Figure~\ref{fig:gCH-p=2-nuispos-tp}. 

A similar behavior occurs when $\nu=\tfrac{27}{32}$,
with the ``bounce'' coinciding with a collision between the (anti-) peakons. 

When $0\leq\nu<\tfrac{27}{32}$, 
there is a collision point but no turning point. 
The collision can occur for either two (anti-) peakons or a peakon and an anti-peakon,
depending on whether $\nu \gtrless \tfrac{9}{16}$ respectively. 
If $\nu>0$ then the separation in the asymptotic past and future is increasing such that 
$\dot\beta_{1,2}\sim \tfrac{1}{6} \alpha_{1,2}(3M^2 +\alpha_{1,2}^2)$ 
and $\alpha_{1,2}^2 \sim 3\big(\sqrt{1/\nu}-1\big)$ are constant,
as shown by the dynamical ODEs \eqref{gCH-aeqn}--\eqref{gCH-beqn}. 
Moreover, in contrast to the case $p=1$, 
$\alpha_{1,2}$ will be non-zero at the time of the collision. 
See Figure~\ref{fig:gCH-p=2-nuispos-cp}. 

But if $\nu=0$, we have $(\alpha_{1,2}/M)^2 = 3(1+B_{1,2})/(3B_{1,2}-1) \geq 3$, 
which requires $B_{1,2}>\tfrac{1}{3}$, and hence $|\beta_{1,2}| < \ln 3$. 
Then the dynamical ODE \eqref{gCH-beqn} shows that 
$|\beta_{1,2}|\to\ln 3$ in a finite time,
and consequently $|\alpha_{1,2}|\to\infty$. 
This behaviour describes a blow up 
in the relative amplitude of a peakon and an anti-peakon, before and after a collision, 
as their separation approaches $|\beta_{1,2}|\to\ln 3$ in a finite time. 
At the collision point, $|\alpha_{1,2}| = \sqrt{3} M$ is non-zero. 
See Figure~\ref{fig:gCH-p=2-nuis0}. 

Finally, when $\nu<0$, 
there is a collision point $\beta_{1,2}=0$, 
and a turning point $0<\beta_{1,2}<\ln 3$ given by the root of 
$48\nu e^{-2|\beta_{1,2}|}(1-e^{-|\beta_{1,2}|}) = (1-3e^{-|\beta_{1,2}|})^3$.
The behaviour in this case is similar to the case $\nu=0$, 
except that the separation between the peakon and the anti-peakon 
increases (before and after the collision) to a maximum at the turning point 
and then decreases to a non-zero limit 
in a finite time such that $|\alpha_{1,2}|\to\infty$. 
See Figure~\ref{fig:gCH-p=2-nuisneg}. 

\begin{figure}[H]
\centering
\begin{subfigure}[t]{.45\textwidth}
\includegraphics[width=\textwidth]{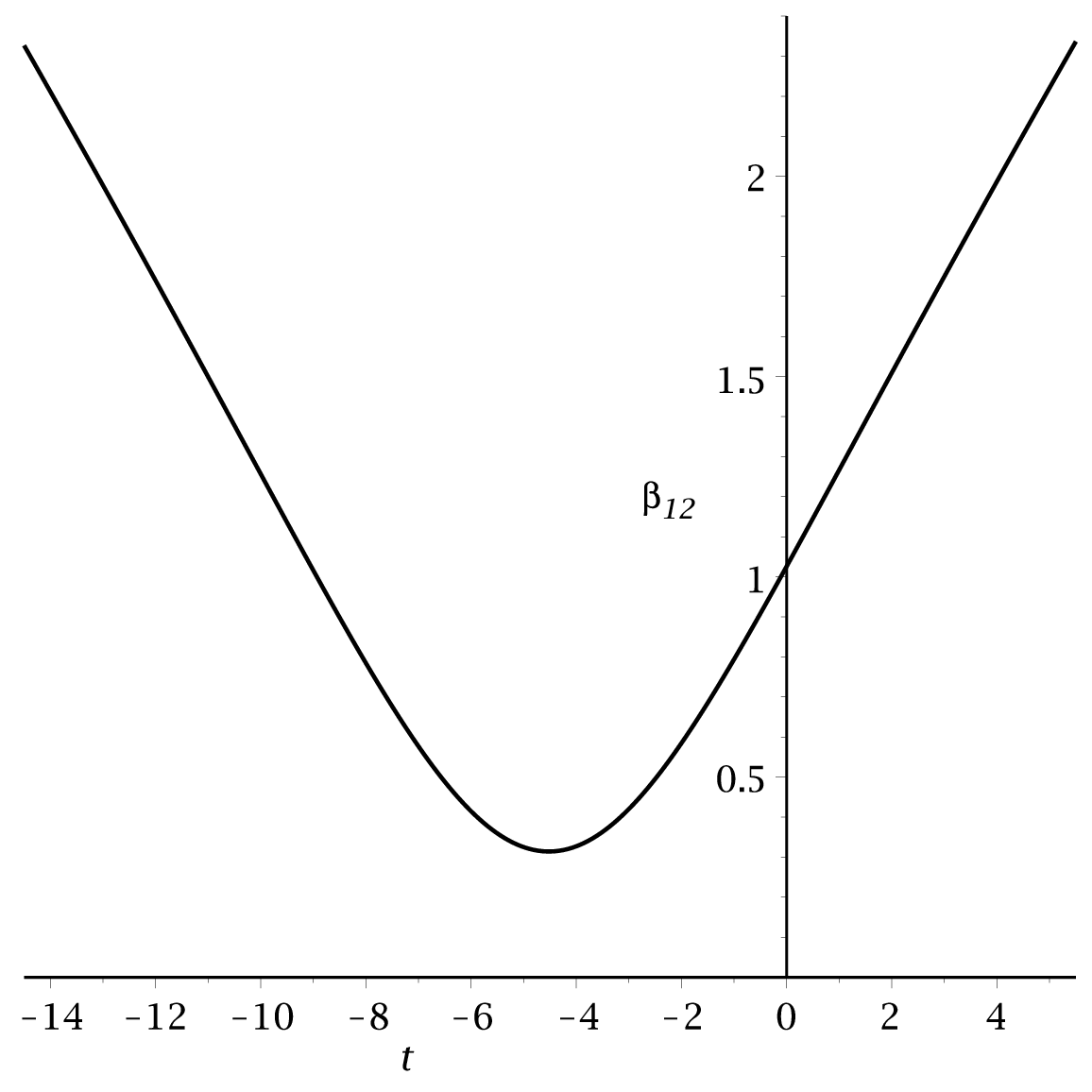}
\captionof{figure}{Separation}
%\label{fig:gCH-p=2-nuispos-tp-position}
\end{subfigure}%
\begin{subfigure}[t]{.45\textwidth}
\includegraphics[width=\textwidth]{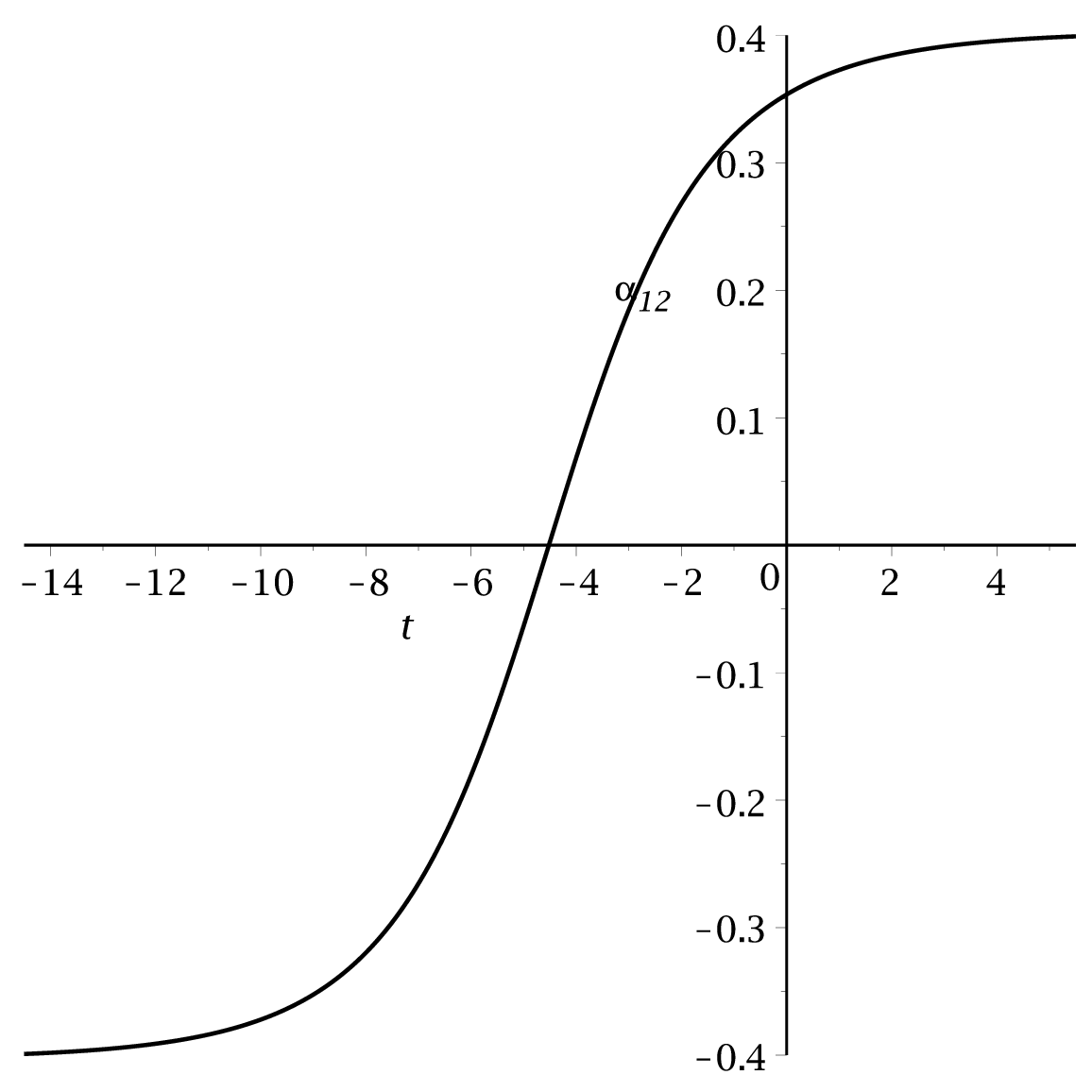}
\captionof{figure}{Relative amplitude}
%\label{fig:gCH-p=2-nuispos-tp-amplitude}
\end{subfigure}
\caption{Relative position and amplitude for $p=2$ gCH 2-peakon interaction when $\tfrac{27}{32}<\nu<1$}
\label{fig:gCH-p=2-nuispos-tp}
\end{figure}

\begin{figure}[H]
\centering
\begin{subfigure}[t]{.45\textwidth}
\includegraphics[width=\textwidth]{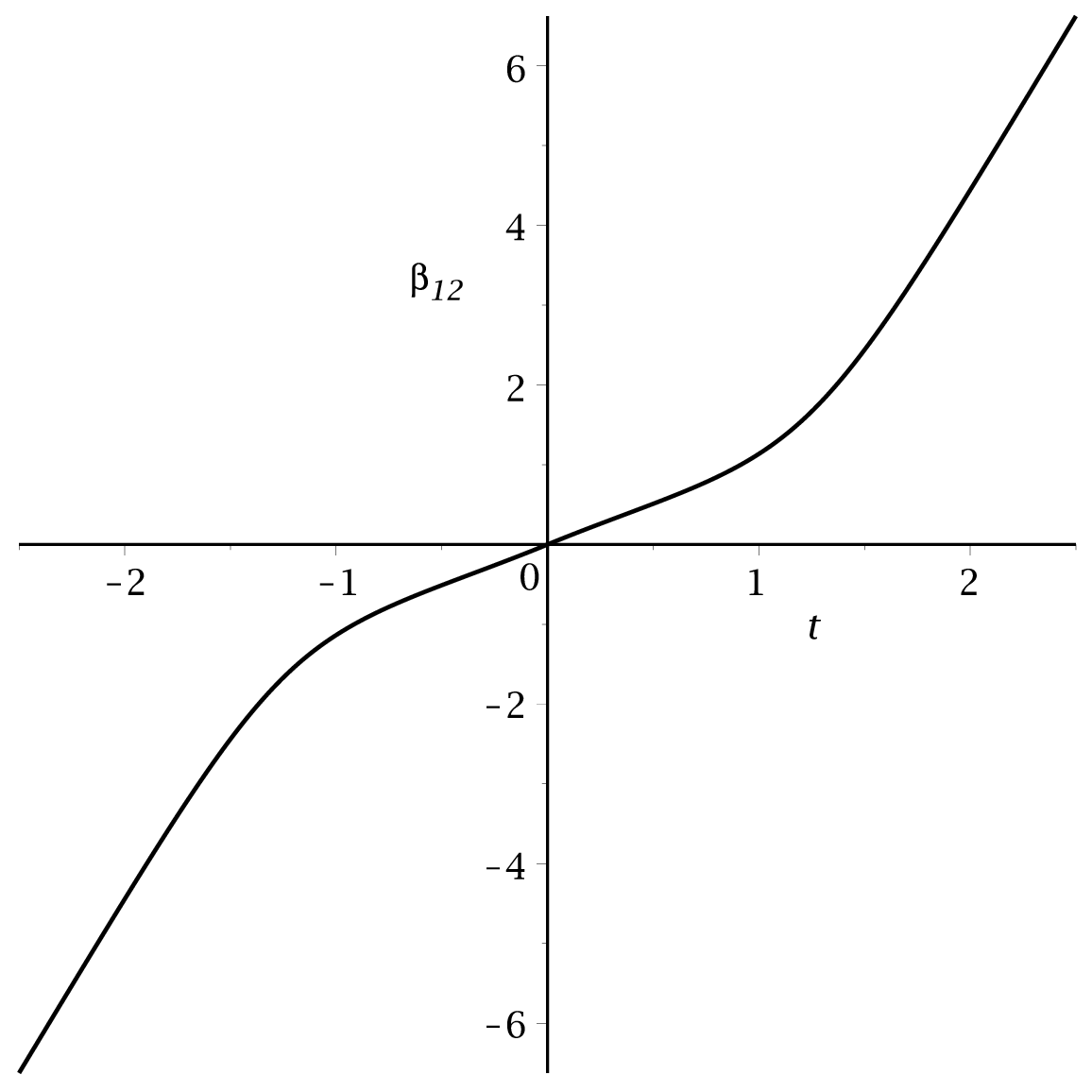}
\captionof{figure}{Separation}
%\label{fig:gCH-p=2-nuispos-cp-position}
\end{subfigure}%
\begin{subfigure}[t]{.45\textwidth}
\includegraphics[width=\textwidth]{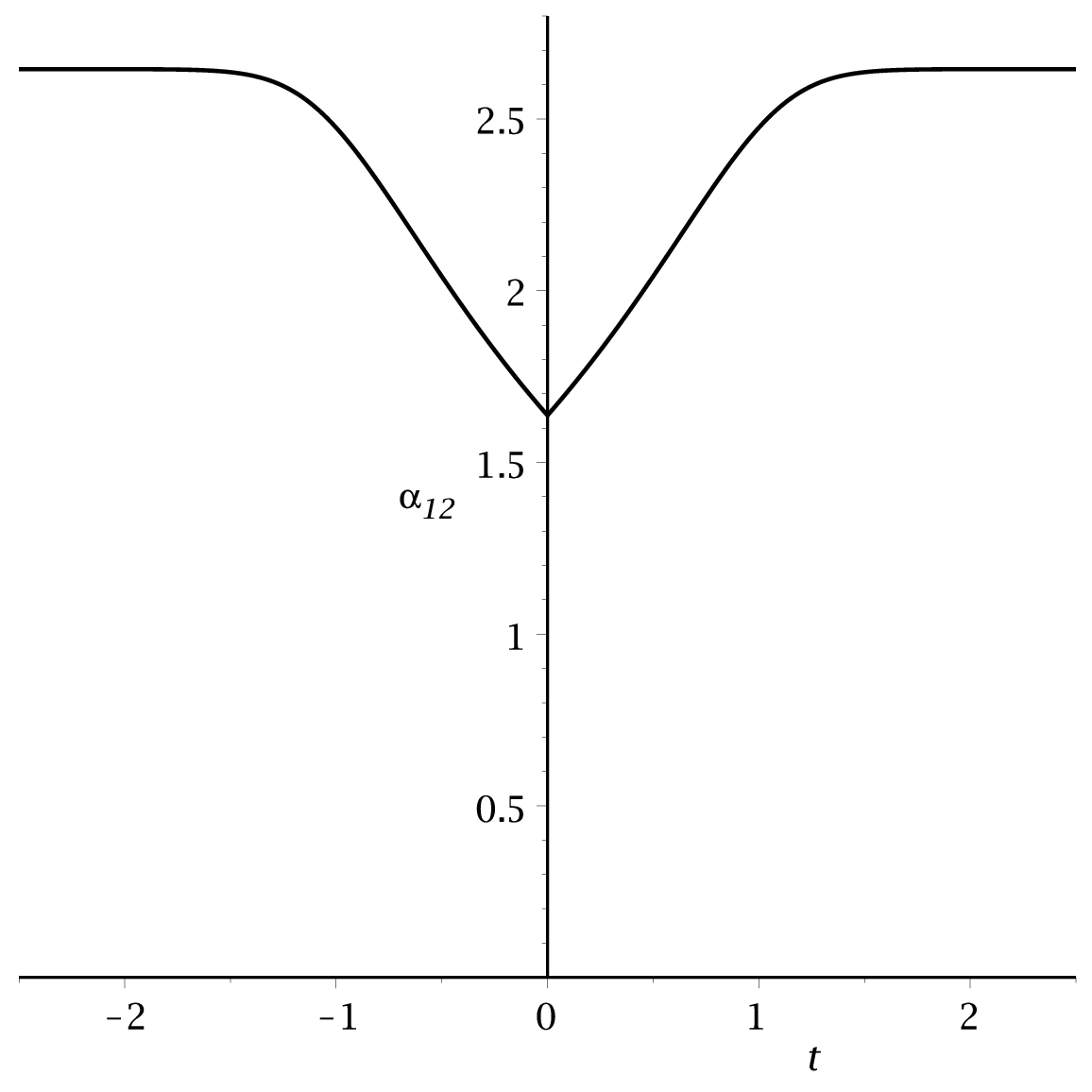}
\captionof{figure}{Relative amplitude}
%\label{fig:gCH-p=2-nuispos-cp-amplitude}
\end{subfigure}
\caption{Relative position and amplitude for $p=2$ gCH 2-peakon interaction when $0<\nu<\tfrac{27}{32}$}
\label{fig:gCH-p=2-nuispos-cp}
\end{figure}

\begin{figure}[H]
\centering
\begin{subfigure}[t]{.45\textwidth}
\includegraphics[width=\textwidth]{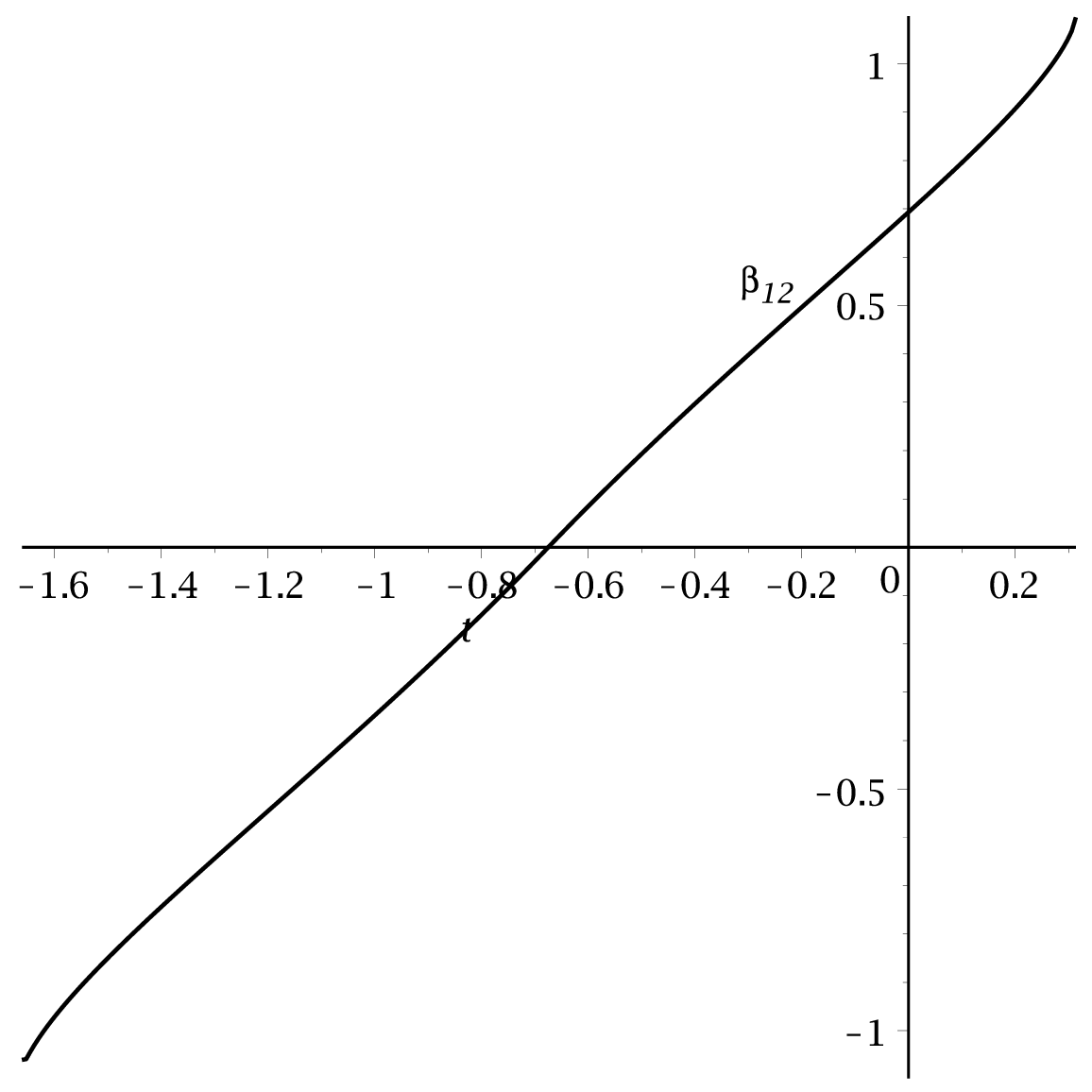}
\captionof{figure}{Separation}
%\label{fig:gCH-p=2-nuis0-position}
\end{subfigure}%
\begin{subfigure}[t]{.45\textwidth}
\includegraphics[width=\textwidth]{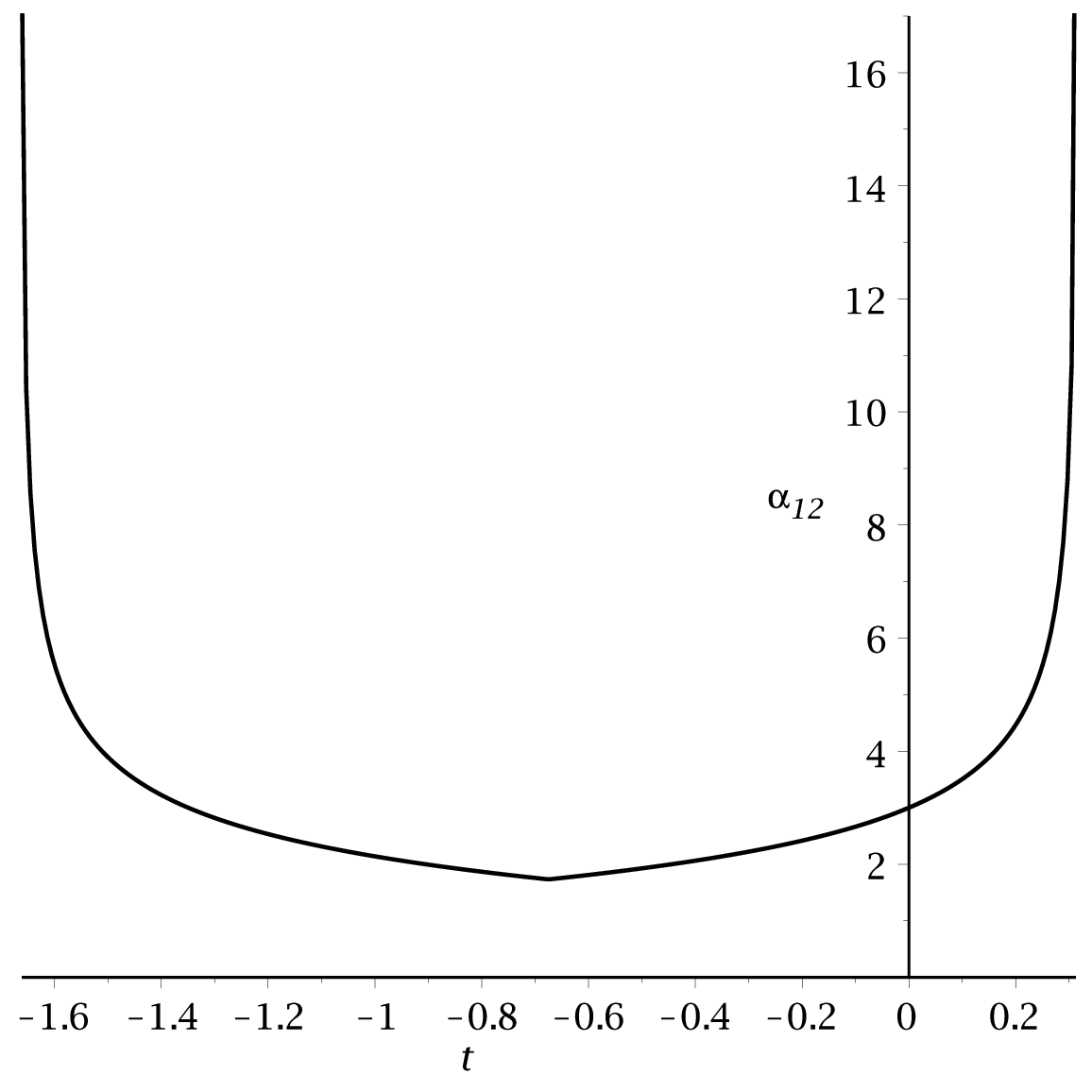}
\captionof{figure}{Relative amplitude}
%\label{fig:gCH-p=2-nuis0-amplitude}
\end{subfigure}
\caption{Relative position and amplitude for $p=2$ gCH 2-peakon interaction when $\nu=0$}
\label{fig:gCH-p=2-nuis0}
\end{figure}

\begin{figure}[H]
\centering
\begin{subfigure}[t]{.45\textwidth}
\includegraphics[width=\textwidth]{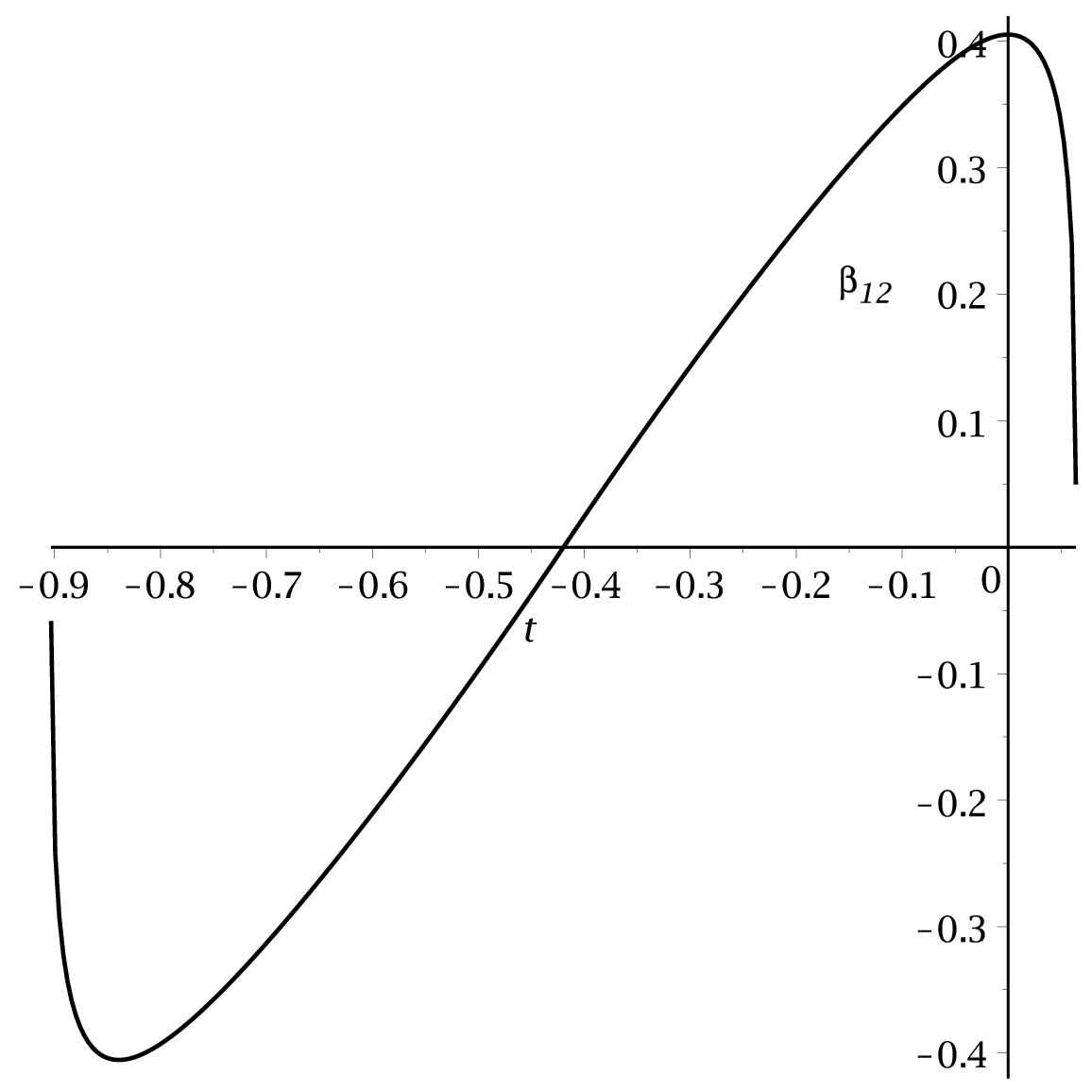}
\captionof{figure}{Separation}
%\label{fig:gCH-p=2-nuisneg-position}
\end{subfigure}%
\begin{subfigure}[t]{.45\textwidth}
\includegraphics[width=\textwidth]{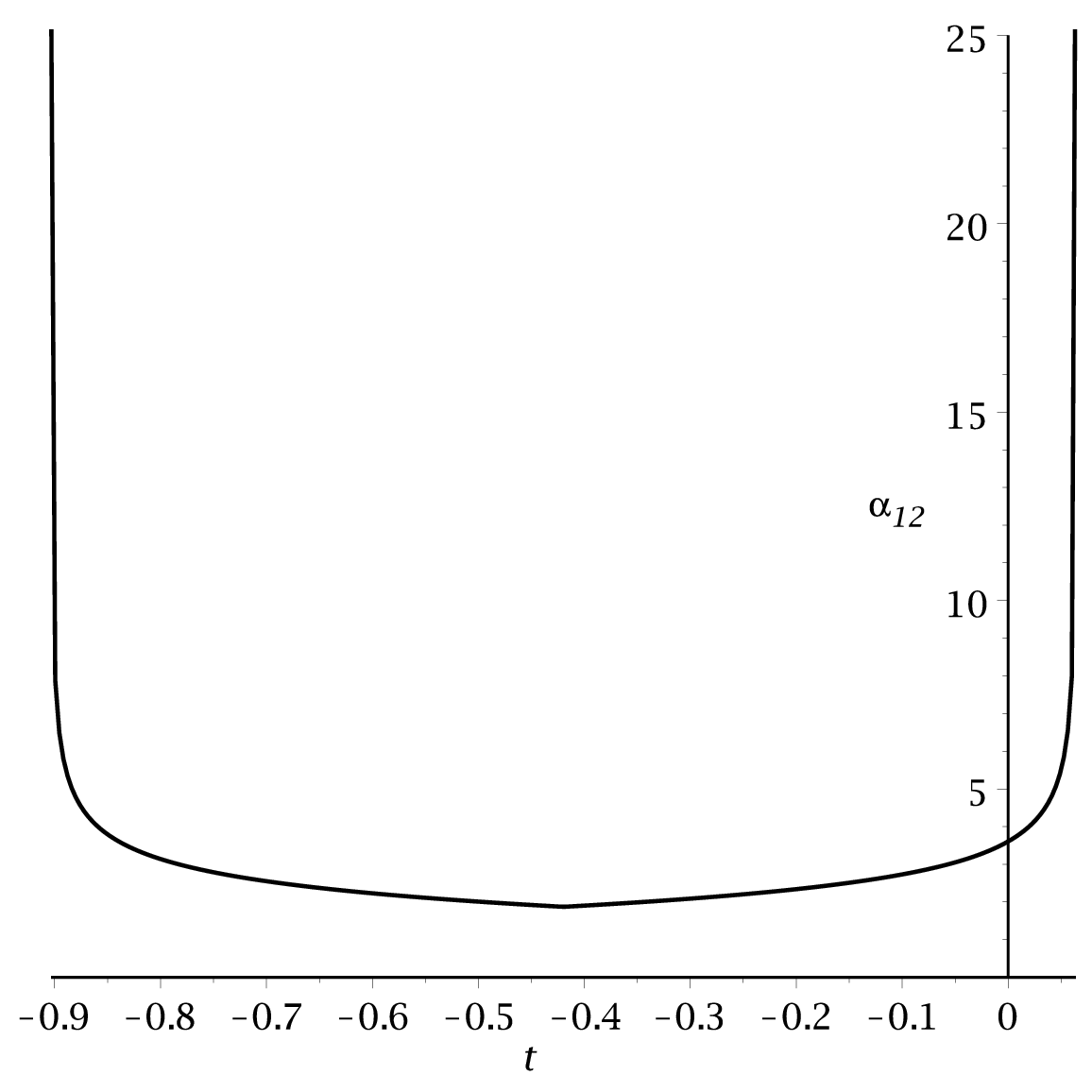}
\captionof{figure}{Relative amplitude}
%\label{fig:gCH-p=2-nuisneg-amplitude}
\end{subfigure}
\caption{Relative position and amplitude for $p=2$ gCH 2-peakon interaction when $\nu<0$}
\label{fig:gCH-p=2-nuisneg}
\end{figure}

These behaviours for $\nu\leq 0$ are strikingly different than 
the blow-up collisions and the elastic ``bounces'' that occur in the ordinary CH equation. 
Further analysis, including explicit solution expressions for $\beta_{1,2}$ and $\alpha_{1,2}$, 
will be given elsewhere.

\subsection{2-peakon solutions of the gmCH equation}

Similarly to the investigation of the gCH equation, 
we will consider the $N=2$ dynamical system in the cases $p=1$ and $p=2$
and examine how the separation between the two (anti-) peakons 
behaves with time $t$. 
Neither of these cases has been extensively explored in the literature to-date. 

The case $p=1$ corresponds to the mCH equation \eqref{FORQ}.
For $N=2$, the dynamical system \eqref{gmCH-multi-peakon} becomes 
\begin{equation}
\dot\alpha_1 = 0 ,
\quad
\dot\alpha_2 = 0 ,
\quad
\dot\beta_1 = \tfrac{2}{3} \alpha_1^2 + 2\alpha_1 \alpha_2 e^{-|\beta_{1,2}|},
\quad
\dot\beta_2 = \tfrac{2}{3} \alpha_2^2 + 2\alpha_1 \alpha_2 e^{-|\beta_{1,2}|}
\end{equation}
where $\beta_{1,2} = \beta_1-\beta_2$ is the separation between 
the two (anti-) peakons or the peakon and the anti-peakon. 
The separation obeys the simple dynamical ODE 
$\dot\beta_{1,2} = \tfrac{2}{3} (\alpha_1^2 -\alpha_2^2) $. 
There are essentially two different types of behaviour, 
depending on the constant of motion 
$\gamma = \tfrac{2}{3} (\alpha_1^2 -\alpha_2^2) $. 

When $\gamma\neq 0$, 
we see that the separation $\beta_{1,2}(t)$ changes linearly in time $t$. 
Hence, there is a collision at a finite time,
and in the asymptotic past and future, 
$\dot\beta_1 \sim \tfrac{2}{3} \alpha_1^2 = c_1$ and $\dot\beta_2 \sim \tfrac{2}{3} \alpha_2^2 = c_2$ 
are the asymptotic speeds of the (anti-) peakons. 
These asymptotic speeds are precisely 
the speeds of the (anti-) peakons in the absence of any interaction,
as shown by the relation \eqref{gmCH-speed}. 
See Figure~\ref{fig:gmCH-p=1-position}. 

In contrast, when $\gamma=0$, 
the separation $\beta_{1,2}(t)$ is time independent. 
This describes two (anti-) peakons, or a peakon and an anti-peakon,
which have equal amplitudes $|\alpha_1|=|\alpha_2|$ 
and equal speeds 
$\dot\beta_1 = \dot\beta_2 = \alpha^2(\tfrac{2}{3} + 2 e^{-|\beta_{1,2}(0)|}) >0$,
where $\alpha = |\alpha_1|=|\alpha_2|$. 

Next, the case $p=2$ corresponds to the gmCH equation \eqref{gmCH},
describing a nonlinear generalization of the mCH equation \eqref{FORQ}.
For $N=2$, 
the dynamical system \eqref{gmCH-multi-peakon} in the case $p=2$ is given by 
\begin{gather}
\dot\alpha_1 = 0 ,
\quad
\dot\beta_1 = \tfrac{8}{15} \alpha_1^2( \alpha_1^2  + 5\alpha_1 \alpha_2 e^{-|\beta_{1,2}|} + 10\alpha_2^2 e^{-2|\beta_{1,2}|} ), 
\\
\dot\alpha_2 = 0 ,
\quad
\dot\beta_2 = \tfrac{8}{15} \alpha_2^2( \alpha_2^2  + 5\alpha_1 \alpha_2 e^{-|\beta_{1,2}|} + 10\alpha_1^2 e^{-2|\beta_{1,2}|} ) . 
\end{gather}
The separation $\beta_{1,2} = \beta_1-\beta_2$ satisfies the dynamical ODE
\begin{equation}\label{separation-p=2}
\dot\beta_{1,2} = \tfrac{8}{15}(\alpha_1^4 -\alpha_2^4) + \tfrac{8}{3} \alpha_1 \alpha_2(\alpha_1^2 -\alpha_2^2) e^{-|\beta_{1,2}|}
= \gamma ( 1+\sigma e^{-|\beta_{1,2}|} )
\end{equation}
where 
\begin{equation}
\gamma = \tfrac{8}{15}(\alpha_1^4 -\alpha_2^4) ,
\quad
\sigma = 5\alpha_1 \alpha_2/(\alpha_1^2 +\alpha_2^2) 
\end{equation}
are constants of motion.
It is straightforward to integrate this ODE to obtain $\beta_{1,2}(t)$ explicitly. 
The qualitative behaviour of $\beta_{1,2}(t)$ depends essentially on both $\gamma$ and $\sigma$. 

First, when $\gamma=0$, 
the amplitudes are equal, $|\alpha_1|=|\alpha_2|$, 
and the separation $\beta_{1,2}(t)$ is time independent,
with the speeds being given by 
$\dot\beta_1 = \dot\beta_2 = \tfrac{8}{15} \alpha^4( 1 + 5 e^{-|\beta_{1,2}(0)|} + 10 e^{-2|\beta_{1,2}(0)|} )>0$, 
where $\alpha = |\alpha_1|=|\alpha_2|$. 
This is the same qualitative behaviour that occurs for $p=1$. 

Next, when $\gamma\neq 0$, 
the behaviour is most easily understood by looking at 
the collision points, defined by $\beta_{1,2}=0$, 
and the turning points, defined by $\dot\beta_{1,2}=0$. 
Note that $\sigma >0$ implies $\alpha_1$ and $\alpha_2$ have the same sign,
and that $\sigma<0$ implies $\alpha_1$ and $\alpha_2$ have opposite signs. 
(The case $\sigma=0$ corresponds to $\alpha_1=0$ or $\alpha_2=0$, which is trivial.)

When $\sigma >-1$, 
we have $1+\sigma e^{-|\beta_{1,2}|} > 0$
which implies there are no turning points. 
Consequently, 
the separation $\beta_{1,2}(t)$ increases in the asymptotic past and future, 
such that asymptotic speeds are $\dot\beta_1 \sim \tfrac{8}{15} \alpha_1^4$
and $\dot\beta_2 \sim \tfrac{8}{15} \alpha_2^4$. 
At a finite time, the separation will be zero. 
This behaviour represents either two (anti-) peakons if $\sigma>0$, 
or a peakon and an anti-peakon if $0>\sigma>-1$, 
that undergo a collision and separate asymptotically to infinity before and after the collision. 
From the amplitude-speed relation \eqref{gmCH-speed},
we see that their asymptotic speeds are precisely 
the speeds of the (anti-) peakons in the absence of any interaction. 
This behaviour is again qualitatively the same as the case $p=1$. 
See Figure~\ref{fig:gmCH-p=2-sigmaispos-position}.

\begin{figure}[H]
\centering
\begin{subfigure}[t]{.45\textwidth}
\includegraphics[width=\textwidth]{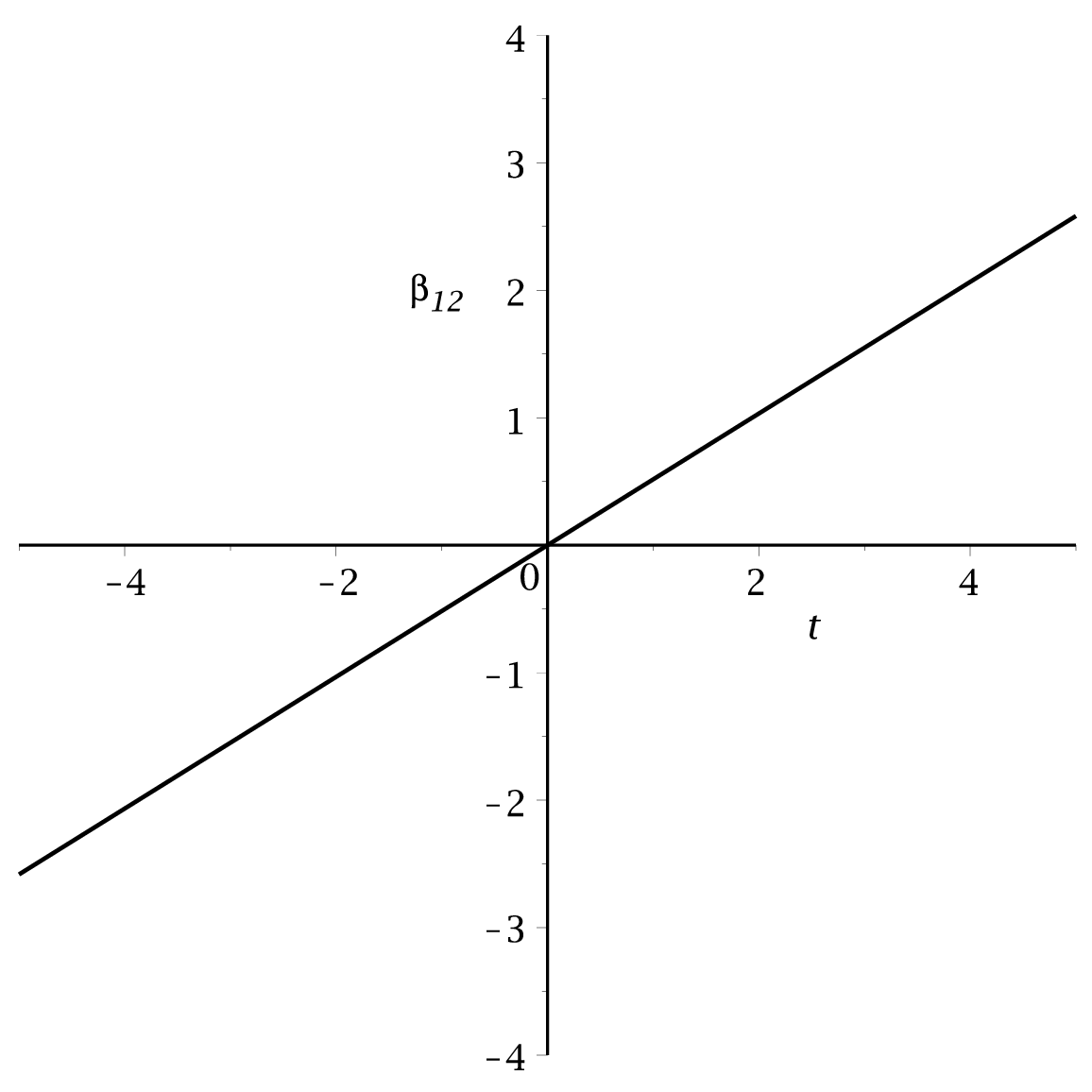}
\captionof{figure}{$p=1$}
\label{fig:gmCH-p=1-position}
\end{subfigure}%
\begin{subfigure}[t]{.45\textwidth}
\includegraphics[width=\textwidth]{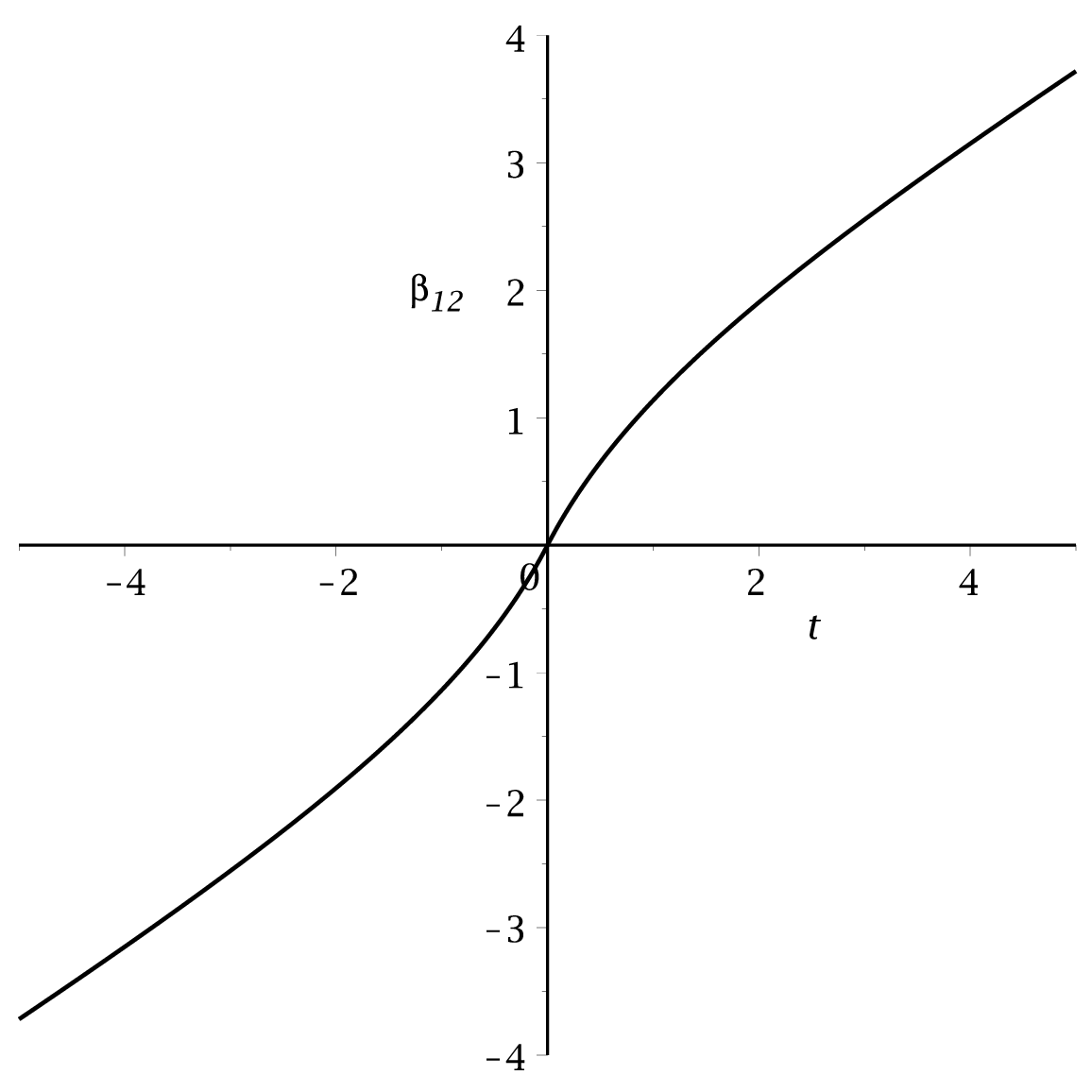}
\captionof{figure}{$p=2$ when $\sigma>0$}
\label{fig:gmCH-p=2-sigmaispos-position}
\end{subfigure}
\caption{Relative position for $p=1,2$ gmCH 2-peakon interaction}
%\label{fig:gmCH-p=1,2}
\end{figure}

When $\sigma\leq -1$, 
there is a turning point given by $|\beta_{1,2}| = \ln|\sigma|$. 
This turning point acts as a critical value for the separation 
between a peakon and an anti-peakon. 
If the initial separation $|\beta_{1,2}(0)|$ is less than $\ln|\sigma|$, 
then the peakon and the anti-peakon form a bound pair in the asymptotic past, 
with their separation asymptotically given by $\ln|\sigma|$. 
In a finite time, the pair undergoes a collapse such that the separation goes to zero
in a collision, 
and then the pair re-forms in the asymptotic future. 
But if the initial separation $|\beta_{1,2}(0)|$ is greater than $\ln|\sigma|$, 
then the peakon and the anti-peakon form a bound pair 
only in the asymptotic past or the asymptotic future,
and this pair breaks apart in the other asymptotic direction,
with the separation going to infinity. 
In the special case when $|\beta_{1,2}(0)|$ is equal to $\ln|\sigma|$, 
the peakon and the anti-peakon form a bound pair with a constant separation given by 
the value $\ln|\sigma|$. 
See Figure~\ref{fig:gmCH-p=2-sigmaisneg}.

\begin{figure}[H]
\centering
\begin{subfigure}[t]{.45\textwidth}
\includegraphics[width=\textwidth]{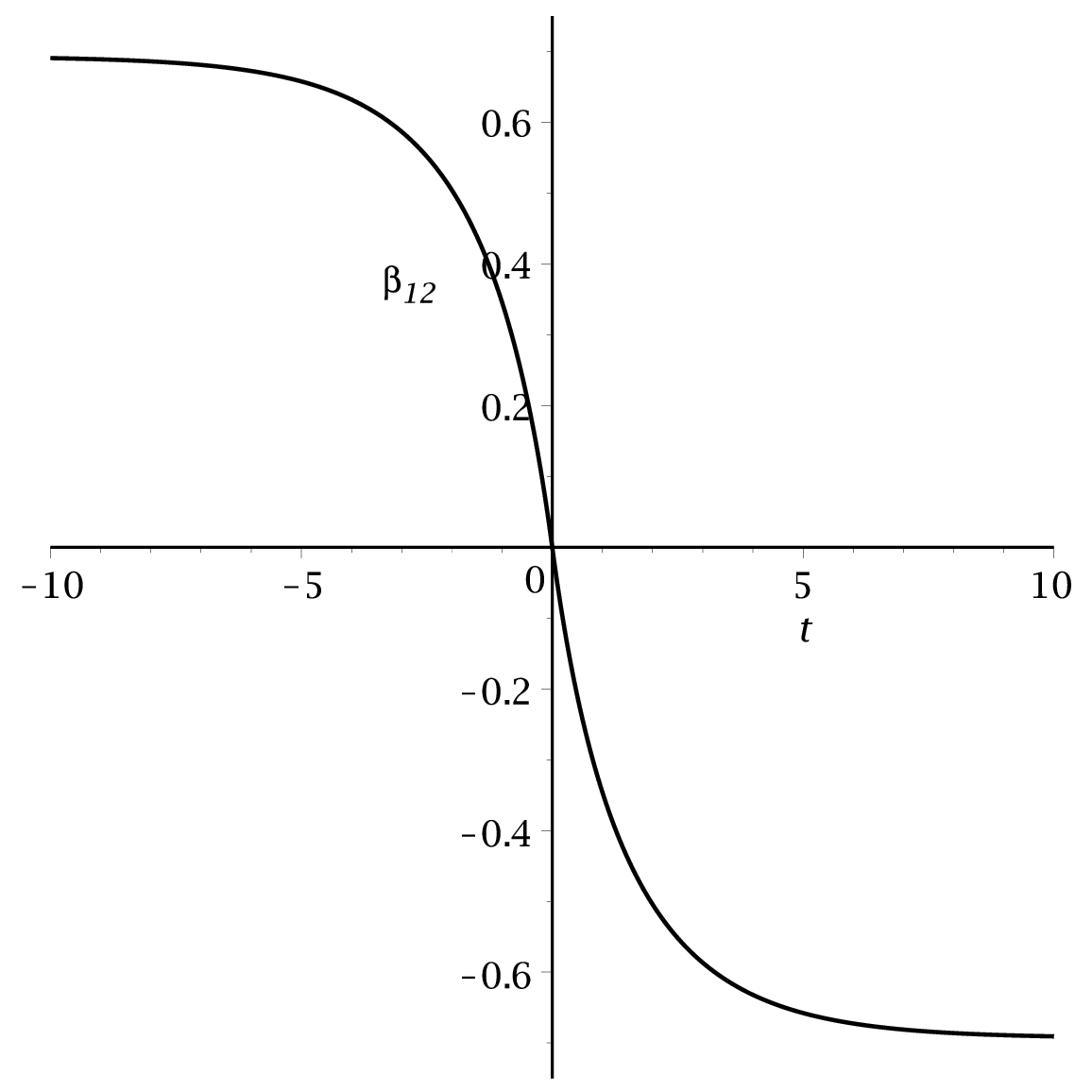}
\captionof{figure}{$|\beta_{1,2}(0)| <\ln|\sigma|$}
%\label{fig:gmCH-p=2-sigmaisneg-position}
\end{subfigure}%
\begin{subfigure}[t]{.45\textwidth}
\includegraphics[width=\textwidth]{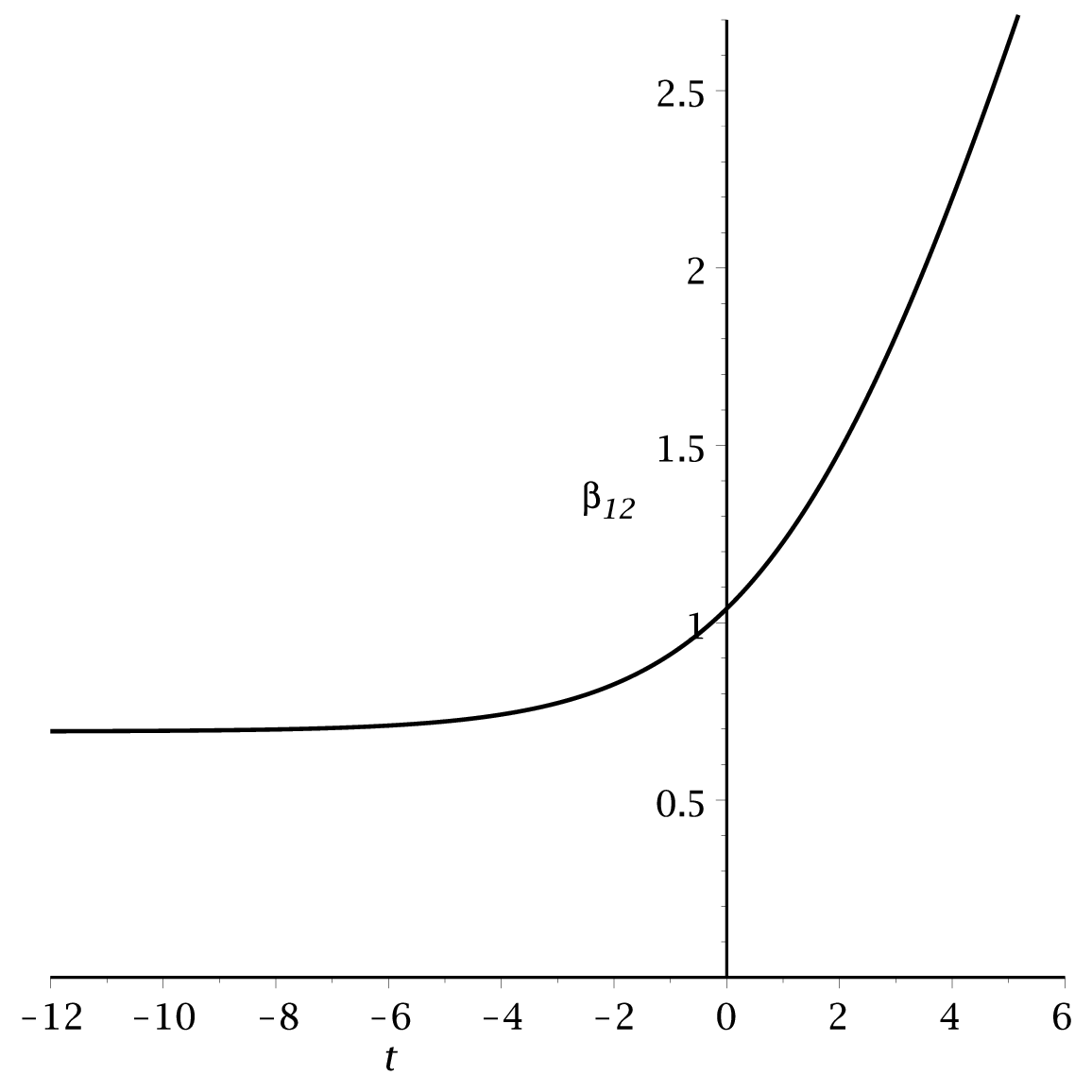}
\captionof{figure}{$|\beta_{1,2}(0)| >\ln|\sigma|$}
%\label{fig:gmCH-p=2-sigmaisneg-position}
\end{subfigure}
\caption{Relative position for $p=2$ gmCH 2-peakon interaction when $\sigma<-1$}
\label{fig:gmCH-p=2-sigmaisneg}
\end{figure}

The formation of a bound pair is a qualitatively novel and striking behaviour 
that has not been seen in any 2-peakon weak solutions of other peakon equations. 
Further analysis of this behaviour, 
including explicit solution expressions for $\beta_{1,2}$ and $\alpha_{1,2}$, 
will be given elsewhere.

We remark that bound pairs have also been observed recently \cite{ChaSzm}
for conservative 2-peakon solutions of the mCH equation.
However, conservative peakons are not weak solutions and instead arise from a different kind of regularization of the mCH equation for distributional solutions. 
Moreover, the dynamics of conservation single peakons of the mCH equation is trivial.
More discussion of these differences can be found in Ref.~\cite{AncKra}.

\section{Conclusions}\label{sec:remarks}

In the $fg$-family \eqref{fg-fam} of nonlinear dispersive wave equations, 
involving two arbitrary functions $f(u,u_x)$ and $g(u,u_x)$,
we have shown that $N$-peakon weak solutions \eqref{multi-peakon}
exist for arbitrary $N\geq 1$.
Neither an integrability structure nor a Hamiltonian structure has been used 
to derive these multi-peakon solutions,
and very likely most of the equations in this family do not possess either type of structure.

We have obtained an explicit dynamical system of ODEs \eqref{ampl_eqn}--\eqref{UV}
for the amplitudes and positions of the $N$ individual peaked waves 
in a general multi-peakon solution \eqref{multi-peakon}. 
For the case $N=1$, in contrast, 
we show that a symmetry condition \eqref{peakon_cond_F}--\eqref{peakon_cond_G}
on $f(u,u_x)$ and $g(u,u_x)$ is necessary for the solution to describe
a single peakon travelling wave \eqref{singlepeakon}.
This condition holds whenever $f(u,u_x)$ is odd in $u_x$ and $g(u,u_x)$ is even in $u_x$. 
Most interestingly, when the condition is not satisfied,
we have shown that the $N=1$ solution instead describes
a generalized dynamical peakon whose amplitude and speed can be time dependent.
Further exploration of dynamical peakons is given in Ref.~\cite{AncRec2018b}. 

As interesting examples of our general results, 
a subfamily of peakon equations that possesses 
the Hamiltonian structure \eqref{Hamil_struc} 
shared by both the CH and FORQ/mCH equations
has been obtained.
We have shown that this Hamiltonian subfamily contains 
a one-parameter nonlinear generalization of the CH equation
and a one-parameter nonlinear generalization of the FORQ/mCH equation, 
as well as a one-parameter multi-peakon equation \eqref{CH-FORQ-fam} that unifies these two generalizations. 
We have derived the single peakon travelling-wave solutions $u=a\exp(-|x-ct|)$ 
for these equations 
and discussed the relation between the properties of 
peakons with $a>0$ and anti-peakons with $a<0$. 

The generalized CH and FORQ/mCH equations involve an arbitrary nonlinearity power $p\geq 1$, 
where the ordinary CH and FORQ/mCH equations correspond to the case $p=1$.  
For both equations, 
we have investigated the effect of higher nonlinearity on (anti-) peakon interactions 
by studying the behaviour of $2$-peakon weak solutions in case $p=2$. 
Qualitatively new behaviours are shown to occur in the interaction between 
a peakon and an anti-peakon. 
Specifically, 
for the $p=2$ generalized mCH equation \eqref{gmCH}, 
the peakon and anti-peakon can form a bound pair 
which has a maximum finite separation in the asymptotic past and future
and which undergoes a collapse at a finite time.
(Stable bound pairs have been seen recently \cite{ChaSzm} for the mCH equation,
but these pairs comprise conservative peakons
which are not weak solutions
and which have trivial dynamics as single peakons \cite{AncKra}.)
In contrast, for the $p=2$ generalized CH equation \eqref{gCH}, 
the peakon and anti-peakon can exhibit a finite time blow-up in amplitude, 
before and after they undergo a collision,
where their separation increases to a finite maximum 
and then decreases to a limiting non-zero value when the blow-up occurs.

The novel behaviours of interactions of peakon and anti-peakon weak solutions
studied here indicate that peakons can exhibit a rich variety of dynamics 
for different multi-peakon equations in the general $fg$-family \eqref{fg-fam},
and that the form of the nonlinearity in these equations has a large impact on 
how peakons can interact.

A study of conservation laws (energy, momentum, $H^1$-norm, etc.)
for the $fg$-family \eqref{fg-fam} has been carried out in recent work \cite{AncRec2018a}. 
There are several interesting directions for future work:
\newline
$\bullet$ 
seek subfamilies of the $fg$-family of multi-peakon equations 
having other types of Hamiltonian structure; 
\newline
$\bullet$ 
explore the conditions under which a Hamiltonian structure will be inherited by 
the dynamical system of ODEs for multi-peakons;
\newline
$\bullet$ 
study the interactions of multi-peakons for high-degree nonlinearities; 
\newline
$\bullet$
understand the difference between peakon weak solutions and
conservative peakon solutions for cubic and higher-degree nonlinearities; 
\newline
$\bullet$ 
study local well-posedness of the $fg$-family of multi-peakon equations;
\newline
$\bullet$ 
for the Cauchy problem, investigate global existence of solutions, 
wave breaking mechanisms, and blow-up times;
\newline
$\bullet$ 
look for new integrable equations in the $fg$-family.

\appendix
\section{Tools from variational calculus}

Three useful tools are the Frechet derivative, the Euler operator, 
and the Helmholtz operator,
which are part of the variational bi-complex \cite{Olv} for differential forms 
in the jet space $J=(x,u,u_x,u_{xx},\ldots)$. 
Differential functions refer to functions on $J$ depending on only finitely many $x$-derivatives of $u$. 

The Frechet derivative is defined by 
\begin{equation}\label{frechet_op}
\delta_P = P\partial_{u} +(D_x P)\partial_{u_x} + (D_x^2P) \partial_{u_{xx}} + \cdots
\end{equation}
where $P$ is an arbitrary differential function. 
Integration by parts on this operator yields 
the adjoint Frechet derivative 
\begin{equation}\label{adjfrechet_op}
\delta_P^* = P\partial_{u} -D_x(P\partial_{u_x}) + D_x^2(P\partial_{u_{xx}}) + \cdots . 
\end{equation}
The Euler operator is given by 
\begin{equation}\label{euler_op}
\E_u = \partial_{u} -D_x\partial_{u_x} + D_x^2 \partial_{u_{xx}} + \cdots . 
\end{equation}
It has the property that it annihilates a differential function $P$ 
iff $P=D_x Q$ is a total $x$-derivative for some differential function $Q$. 
Similarly, 
the Helmholtz operator 
\begin{equation}\label{helmholtz_op}
\delta_P - \delta_P^* = P(\partial_{u} - \E_u) + (D_x P)(\partial_{u_x} +\E_u^{(1)}) +  (D_x^2 P)(\partial_{u_{xx}} -\E_u^{(2)}) +\cdots
\end{equation}
has the property that it annihilates a differential function $R$
with $P$ being an arbitrary differential function,  
iff $R= \E_u(Q)$ is an Euler-Lagrange expression for some differential function $Q$. 
Here 
\begin{equation}\label{higher_euler_op}
\E_u^{(k)} = \partial_{u} -\tbinom{k+1}{1}D_x\partial_{u_x} + \tbinom{k+2}{2}D_x^2 \partial_{u_{xx}} + \cdots,
\quad
k=1,2,\ldots
\end{equation}
denote the higher Euler operators. 

A useful identity relating the Frechet derivative and the Euler operator is given by 
\begin{equation}\label{frechet_euler_ident}
\delta_R P = R\,\E_u(P) + D_x\Theta,
\quad
\Theta = R\,\E_u^{(1)}(P) +  (D_xR)\,\E_u^{(2)}(P) +\cdots .
\end{equation}

\subsection{Proof of Proposition~\ref{prop:fg-Hamil}}

Using the Euler operator and the Helmholtz operator, 
it is straightforward to see that a Hamiltonian structure of the form \eqref{Hamil_struc} 
exists for the $fg$-family \eqref{fg-fam} of peakon equations 
iff the functions $f(u,u_x)$ and $g(u,u_x)$ 
satisfy the conditions
\begin{equation}\label{euler_eqn}
\E_u( f(u,u_x)(u-u_{xx}) )= 0
\end{equation}
and
\begin{subequations}\label{helmholtz_eqns}
\begin{gather}
\partial_{u}(g(u,u_x)(u-u_{xx}) +A) - \E_{u}(g(u,u_x)(u-u_{xx}) +A)  =0, 
\label{helmholtz_eqn1}
\\
\partial_{u_x}(g(u,u_x)(u-u_{xx}) +A)  +\E_{u_{x}}^{(1)}(g(u,u_x)(u-u_{xx}) +A) =0, 
\label{helmholtz_eqn2}
\\
\partial_{u_{xx}}(g(u,u_x)(u-u_{xx}) +A) -\E_{u_{xx}}^{(2)}(g(u,u_x)(u-u_{xx}) +A) =0
\label{helmholtz_eqn3}
\end{gather}
\end{subequations}
where 
\begin{equation}\label{fg_hamil_grad}
A= D_x{}^{-1}( f(u,u_x)(u-u_{xx}) ) .
\end{equation}

Both of the conditions \eqref{euler_eqn} and \eqref{helmholtz_eqns}
must hold identically in jet space, 
and hence each one will split with respect to all $x$-derivatives of $u$ 
that do not appear in the functions $f(u,u_x)$ and $g(u,u_x)$. 

From condition \eqref{euler_eqn}, 
we get two equations
\begin{equation}
(u_x(u_xf_u+uf_{u_x})-uf)_u=0,
\quad
(u_x(u_xf_u+uf_{u_x})-uf)_{u_x}=0 ,
\end{equation}
which comprise an overdetermined linear system for $f(u,u_x)$.
Integration with respect to $u$ and $u_x$ yields 
$uf - u_x(u_xf_u+uf_{u_x}) =f_0$, where $f_0$ is an arbitrary constant. 
This is a linear first-order PDE whose general solution is expression \eqref{f_sol}. 
We substitute this expression into equation \eqref{fg_hamil_grad}
and note 
\begin{equation}
u_x m f_1(u^2-u_x^2) = \tfrac{1}{2} D_x F_1(u^2-u_x^2), 
\quad
\frac{um}{u^2-u_x^2} = D_x \Big( \tfrac{1}{2} \ln\Big(\frac{u-u_x}{u+u_x}\Big) + x \Big)
\end{equation}
through using the relation $u_xm = \tfrac{1}{2}D_x(u^2-u_x^2)$. 
Hence, we obtain 
\begin{equation}\label{A_sol}
A = \tfrac{1}{2} F_1(u^2-u_x^2) + \tfrac{1}{2} f_0 \ln\Big(\frac{u-u_x}{u+u_x}\Big) + f_0 x . 
\end{equation}
Then we substitute this expression into the second condition \eqref{helmholtz_eqns},
which reduces to the equations
$0= D_x((u-u_{xx})g_{u_x} +A_{u_x}) + D_x^2 g$
and 
$0= (u-u_{xx})g_{u_x} +A_{u_x})  -2D_x g$. 
Clearly, the first equation is a differential consequence of the second equation. 
Expanding the second equation, 
we get 
\begin{equation}
u_xg_u+ug_{u_x} = u_x F_1' + f_0\frac{u}{u^2-u_x^2} 
\end{equation}
which is a linear first-order PDE for $g(u,u_x)$. 
Its general solution is expression \eqref{g_sol}. 

Hence, the first part of Proposition~\ref{prop:fg-Hamil} has been established. 
To prove the second part, 
we need to invert the relation $gm+A= \E_u(B)$ to find the Hamiltonian density $B$, 
where 
\begin{equation}\label{euler_B}
gm+A = 
\tfrac{1}{2} F_1(u^2-u_x^2) + (u f_1(u^2-u_x^2) +g_1(u^2-u_x^2))m
+ f_0\Big( \frac{u_xm}{u^2-u_x^2} + \tfrac{1}{2} \ln\Big(\frac{u-u_x}{u+u_x}\Big) + x \Big)
\end{equation}
is given by expressions \eqref{A_sol} and \eqref{g_sol}. 
The form of the Hamiltonian densities \eqref{CH_hamil} and \eqref{FORQ_hamil}
for the CH and FORQ/mCH equations 
suggests seeking $B= m C(u,u_x)$. 

We expand $\E_u(B)= \E_u((u-u_{xx})C(u,u_x))$, 
equate it to expression \eqref{euler_B}, 
and split with respect to $m$. 
This yields a system of two linear equations. 
By combining these equations, we get
\begin{align}
( uC-u_x(u_xC_u-uC_{u_x}) )_u
& = -uu_x f_1 -u_x g_1 - f_0 \frac{u_x^2}{u^2-u_x^2} , 
\\
( uC-u_x(u_xC_u-uC_{u_x}) )_{u_x}
& = 
u^2 f_1 + \tfrac{1}{2} F_1 +u g_1 
+ f_0 \frac{uu_x}{u^2-u_x^2} 
+\tfrac{1}{2}f_0 \ln\Big(\frac{u-u_x}{u+u_x}\Big) +f_0 x . 
\end{align}
This system is straightforward to integrate. 
Its general solution for $C$ is given by 
\begin{equation}
C = u_x C_1(u^2-u_x^2) + \frac{C_0 u}{u^2-u_x^2} 
+ \tfrac{1}{2} F_1(u^2-u_x^2) + \tfrac{1}{2}\frac{uG_1(u^2-u_x^2)}{u^2-u_x^2}
+ \tfrac{1}{2}f_0 \ln\Big(\frac{u-u_x}{u+u_x}\Big) +f_0 x 
\end{equation}
where $C_1$ is an arbitrary function of $u^2-u_x^2$, and $C_0$ is an arbitrary constant. 
Then $B=mC$ consists of the Hamiltonian density \eqref{fg_Hamil} 
plus a total $x$-derivative 
$D_x\Big( \tfrac{1}{2} \tilde C_1(u^2-u_x^2) + \tfrac{1}{2} C_0 \ln\Big(\dfrac{u-u_x}{u+u_x}\Big) + C_0 x \Big)$, 
with $\tilde C_1{}' = C_1$. 

This completes the proof of Proposition~\ref{prop:fg-Hamil}.

\end{document}